\documentclass[10pt]{article}
\usepackage[nofonts]{dgleich-common}
\usepackage[T1]{fontenc}

\RequirePackage[scaled=0.8]{helvet}%
\RequirePackage[scaled=0.75]{beramono}%
\def\lstbasicfont{\fontfamily{fvm}\selectfont}
\makeatletter  
\setboolean{@dgleich@sanssmallcaps}{true}
\makeatother
\usepackage{xcolor}
\definecolor{firebrick}{HTML}{871a1a}
\definecolor{metro_teal}{HTML}{23373b}
\definecolor{light_teal}{HTML}{7E9AA1}
\usepackage{amsmath}
\usepackage{amsthm}
\usepackage{amssymb}
\usepackage{bm}
\usepackage{graphicx}
\usepackage{authblk}
\graphicspath{{../fig/}{fig/}{assets/}{../assets/}}
\usepackage{substr}

\usenatbib
\usehyperref

    \usepackage[capitalise]{cleveref}
    \usepackage{etoolbox}
    
    \usepackage{graphicx}
    
    \newtoggle{preprint}
    \toggletrue{preprint}

    \newcommand{\articleonly}[2]{
        \iftoggle{preprint}{#2}{#1}%
    }
    \newcommand{\dropcap}[1]{#1}

    \newtoggle{comments}
    \toggletrue{comments}    

\newcommand{\mA}[0]{\mathbf{A}}
\newcommand{\mD}[0]{\mathbf{D}}
\newcommand{\mI}[0]{\mathbf{I}}
\newcommand{\mL}[0]{\mathbf{L}}
\newcommand{\mE}[0]{\mathbf{E}}
\newcommand{\mG}[0]{\mathbf{G}}

\newcommand{\mJ}[0]{\mathbf{J}}

\newcommand{\mGamma}[0]{\bm{\Gamma}}
\newcommand{\mSigma}[0]{\bm{\Sigma}}
\newcommand{\mDelta}[0]{\bm{\Delta}}
\newcommand{\mLambda}[0]{\bm{\Lambda}}

\newcommand{\mS}[0]{\mathbf{S}}

\newcommand{\mM}[0]{\mathbf{M}}

\newcommand{\vs}[0]{\mathbf{s}}
\newcommand{\vv}[0]{\mathbf{v}}

\newcommand{\vgamma}[0]{\bm{\gamma}}

\newcommand{\ve}[0]{\mathbf{e}}

\newcommand{\diag}[1]{\textnormal{diag} #1}
\newcommand{\vf}[0]{\mathbf{f}}

\newcommand{\vk}[0]{\mathbf{k}}
\newcommand{\vzero}[0]{\mathbf{0}}
\newcommand{\vbeta}[0]{\bm{\beta}}

\newcommand{\vd}[0]{\mathbf{d}}

\newcommand{\cL}[0]{\mathcal{L}}

\renewcommand{\t}[0]{{(t)}}
\newcommand{\tplusone}[0]{{(t+1)}}

\newcommand{\dds}[1]{\frac{\partial #1}{\partial \vs}}

\newtheorem{thm}{Theorem}

\newtheorem{lm}{Lemma}
\newtheorem{cor}{Corrolary}

\newtheorem{lemma*}{Lemma}

\newtheorem{prop*}{Proposition}

\theoremstyle{remark}

\newcommand\E[0]{\mathbb{E}}
\newcommand\R[0]{\mathbb{R}}

\newcommand{\argmax}{\operatornamewithlimits{argmax}}

\newcommand\prob[0]{\mathbb{P}}

\crefname{thm}{Theorem}{Theorems}
\Crefname{thm}{Theorem}{Theorems}
\crefname{lm}{Lemma}{Lemmas}
\Crefname{lm}{Lemma}{Lemmata}

\Crefname{clm}{Claim}{Claims}
\crefname{clm}{Claim}{Claims}

\Crefname{cor}{Corollary}{Corollaries}
\crefname{cor}{Corollary}{Corollaries}

\crefname{example}{example}{examples}
\Crefname{example}{Example}{Examples}
\Crefname{table}{Table}{Tables}

\Crefname{equation}{Equation}{Equations}
\crefname{equation}{Eq.\!}{Eqs.\!}
\Crefname{figure}{Figure}{Figures}
\crefname{figure}{Fig.\!}{Figs.\!}

\newcommand{\rev}[1]{\textcolor{black}{#1}}
\newcommand{\revtwo}[1]{\textcolor{black}{#1}}

\begin{document}

\title{Emergence of Hierarchy in Networked Endorsement Dynamics}
\author{Mari Kawakatsu, Philip S. Chodrow, Nicole Eikmeier, and Daniel B. Larremore}

\maketitle


\marginnote[450pt]{\fontsize{7}{9}\selectfont
Mari Kawakatsu \\ Princeton University\\
\texttt{mk28@math.princeton.edu} \\ 
\noindent Philip S. Chodrow \\ 
Massachusetts Institute of Technology \\  
University of California, Los Angeles\\ 
\texttt{phil@math.ucla.edu} \\ 
\noindent Nicole Eikmeier \\ 
Grinnell College\\
\texttt{eikmeier@grinnell.edu} \\ 
\noindent Daniel B. Larremore \\ 
University of Colorado, Boulder \\
\texttt{daniel.larremore@colorado.edu}\\
\vspace{1em}
\noindent MK, PSC, and NE contributed equally to this research. 
}

\begin{abstract}
    \articleonly{}{\vspace{-4em}}
Many social and biological systems are characterized by enduring hierarchies, including those organized around prestige in academia, dominance in animal groups, and desirability in online dating. 
Despite their ubiquity, the general mechanisms that explain the creation and endurance of such hierarchies are not well understood. 
We introduce a generative model for the dynamics of hierarchies using time-varying networks in which new links are formed based on the preferences of nodes in the current network and old links are forgotten over time. 
The model produces a range of hierarchical structures, ranging from egalitarianism to bistable hierarchies, and we derive critical points that separate these regimes in the limit of long system memory. 
\rev{Importantly},
our model supports statistical inference, allowing for a principled comparison of generative mechanisms using data. 
We apply the model to study hierarchical structures in empirical data on hiring patterns among mathematicians, dominance relations among parakeets, and friendships among members of a fraternity,
observing several persistent patterns as well as interpretable differences in the generative mechanisms favored by each. 
Our work contributes to the growing literature on statistically grounded models of time-varying networks.

\end{abstract}

\articleonly{}
{
\section*{Introduction}
\label{sec:intro}
}

\articleonly{\dropcap{H}}{H}ierarchies---stable sets of dominance relationships among~individuals~\cite{fushing2011ranking,hobson2018strategic,hobson2015social}---structure many human and animal societies. 
Among animals, hiearchical rank may determine access to~resources such as food, grooming, and reproduction \cite{holekamp2016aggression}.
Among humans, rank shapes the epistemic capital and employment prospects of researchers \cite{clauset2015hierarchy,morgan2018prestige}, susceptibility of adolescents to bullying  \cite{garandeau2014inequality}, messaging patterns in online dating~\cite{bruch2018aspirational}, and influence in group decision-making  \cite{cheng2014toward}.

A central question concerns how enduring hierarchies shape and are shaped by interactions between individuals. 
Empirical studies have indicated the presence of a \emph{winner effect}: an individual who participates in a favorable interaction, such as winning a fight or receiving an endorsement, increases their likelihood of being favored in future interactions \cite{chase1994aggressive,hogeweg1983ontogeny}. 
Both theoretical work \cite{bonabeau1995phase,bonabeau1996mathematical,hemelrijk1999individual,ben2005dynamics,miyaguchi2020piecewise, posfai2018talent, hickey2019self, vehrencamp1983model, sanchez2018practical} and controlled experiments in humans \cite{salganik2006experimental} suggest that winner effects are sufficient (though not necessary) to form stable hierarchies. 
Mechanistic explanations of winner effects vary. 
A common approach postulates that each individual possesses an \emph{intrinsic strength}, which may depend on factors such as size, skill, or aggression levels. 
For instance, physiological mechanisms such as changes in hormone levels following confrontational interactions \cite{mehta2015dual} can alter an individual's strength, causing the strong to get stronger. 

However, intrinsic strengths are not necessary to produce winner effects. 
If a politician endorses a rival candidate, the latter does not become intrinsically more fit for office; instead, the endorsee builds support for their candidacy that may lead to future endorsements. 
The fame of the endorser is key: the better-known the endorser, the more valuable the endorsement. 
We refer to such prestige by proxy as \emph{transitive prestige}. 
Since transitive prestige enables hierarchical rank to flow through interactions between individuals, networks provide a natural lens through which to study its role.
Recent empirical studies have emphasized the networked nature of hierarchy in biological and social groups \cite{shizuka2015network, ball2013friendship, pinter2014dynamics, hobson2015social, hobson2018strategic}. 
Several theoretical studies \cite{konig2011network, konig2014nestedness, bardoscia2013social, krause2013spontaneous} have also investigated reinforcing hierarchy using time-varying network models called \emph{adaptive networks} \cite{sayama2013modeling, porter2020nonlinearity+}. In this class of models, edges, representing interactions, evolve in response to node states and vice versa. 
Edges tend to accrue to important or highly central nodes, leading to self-reinforcing hierarchical network structures. 
Despite their recent uses, adaptive networks are often difficult to analyze analytically or compare to empirical data. 

We present a novel and flexible adaptive network model of social hierarchy that addresses these challenges. 
Winner effects in our model are driven entirely by social reinforcement rather than intrinsic strengths. 
We allow arbitrary matrix functions to determine rank or prestige of nodes in the network, and introduce parameters governing the behavior of individuals in response to rank.  
\rev{A key feature of our model is that it}
is amenable both to mathematical analysis and to statistical inference. 
We analytically characterize a critical transition separating egalitarian and hierarchical model states for several choices of ranking function. 
We also explore hierarchical patterns in four biological and social data sets, using our model to perform principled selection between competing ranking methods in each data set and highlight persistent macroscopic patterns. 
We conclude with a discussion of potential model extensions and connections to recent work on centrality in temporal networks. 
\section*{Modeling Emergent Hierarchy}
\label{sec:model}

In our adaptive network model, new directed edges are formed based on existing, node-based hierarchy, after which they decay over time. 
We conceptualize a directed edge $i\to j$ as an \emph{endorsement}, in which $i$ affirms that $j$ is fit, prestigious, or otherwise of high quality. 
For example, endorsements could capture contests won by $j$ over $i$, retweets of $j$ by $i$, or comparisons in which a third party ranks $j$ above $i$. 
We collect endorsements in a weighted directed network on $n$ nodes summarized by its adjacency matrix $\mA\in \R^{n\times n}$, where entry $a_{ij}$ is the weighted number of interactions $i\to j$. 
The matrix $\mA$ evolves in discrete time via the iteration 
\begin{equation}
	\mA\tplusone = \lambda \mA\t + (1-\lambda)\mDelta\t\;. \label{eq:dynamics}
\end{equation}
Here, the \textit{update matrix} $\mDelta\t$ contains new endorsements at time $t$. 
The \emph{memory parameter} $\lambda \in [0,1]$ represents the rate with which memories of old endorsements decay; the smaller the value of $\lambda$, the more quickly previous endorsements are forgotten. 

The new endorsements in $\mDelta\t$ depend on previous endorsements through a ranking of the $n$ nodes, which we call the \emph{score vector} (or simply \emph{score}) $\vs \in \R^n$. 
The score vector is the output of a \textit{score function} $\sigma: \mA \mapsto \vs \in \mathbb{R}^n$, which may be any rule that assigns a real number to each node. 

\rev{
	We consider three score functions chosen for analytical tractability and relevance in applications.
	Let $\mD^{\text{in}}$ and $\mD^{\text{out}}$ be diagonal matrices whose entries are the weighted in- and out-degrees of the network, i.e., $\mD^{\text{in}}_{ii} = \sum_{j} \mA_{ij}$ and $\mD^{\text{out}}_{ii} = \sum_{j} \mA_{ji}$.
	First, the Root-Degree score is the square root of the in-degree---the weighted number endorsements---of each node $i$, defined as $s_i = \sqrt{\mD^{\text{in}}_{ii}}$.
	The Root-Degree score function does not model transitive prestige, since only the number of endorsements is considered, not the prestige of the agents from which they come. 
	Second, PageRank~\cite{brin1998anatomy} is a recursive notion of rank in which high-rank nodes are those whose endorsers are numerous, and themselves  high rank. 
	The foundational algorithm used by Google in ranking webpages, PageRank computes a value for each node interpretable as the proportion of time that a random surfer following the network of endorsements would spend on that node. 
	We define PageRank score $\vs$ as the PageRank vector of $\mA^T$, which is the unique solution to the system  
	\begin{align}
		\left[\alpha_p \mA^T(\mD^{\text{out}})^{-1} + (1-\alpha_p)n^{-1}\ve\,\ve^T\,\right]\vs = \vs
	\end{align}
	up to scalar multiplication. 
	Here, $\alpha_p \in [0, 1]$ is the so-called teleportation parameter, for which we use the customary value $\alpha_p = 0.85$. 
	We normalize the PageRank vector so that $\ve^T\vs = n$, where $\ve$ is the vector of ones.
}
\rev{
	Finally, SpringRank~\cite{debacco2018springrank} is another recursive definition of rank in which endorsers are ranked one unit below endorsees, and disagreements are resolved using an analogy to a physical system of springs: the ranking of nodes minimizes the total energy of the system.
	Mathematically, the SpringRank score $\vs$ is the unique solution to the linear system \cite{debacco2018springrank}
	\begin{align}
		\left[\mD^{\text{in}} + \mD^{\text{out}} - (\mA + \mA^T) + \alpha_s \mI\right]\vs 
		= 
		\left[ \mD^{\text{in}} - \mD^{\text{out}}\right] \ve\;, \label{eq:springrank}
	\end{align} 
	with the identity matrix $\mI$ and a regularization parameter $\alpha_s>0$ which ensures the uniqueness of $\vs$. 
	Unlike the Root-Degree score, both PageRank and SpringRank scores model transitive prestige: the impact of an endorsement depends on the prestige of the endorser.  
	These three score functions can all be interpreted as rankings or centrality measures, although \revtwo{this property is not required of score functions} in our model. 
}

Given score vector $\vs$, new endorsements $\mDelta$ are chosen using a random utility model, a standard \revtwo{framework} in discrete choice theory which has recently been applied in models of growing networks \cite{overgoor2019choosing}.
At time step $t$, node $i$ is selected uniformly at random. 
We suppose that endorsing $j$ has utility $u_{ij}(\vs)$ for~$i$, which depends on the current scores. 
In this work, we focus on utilities of the functional form 
\begin{align}
	u_{ij}(\vs) = \beta_1 s_j + \beta_2(s_i - s_j)^2\;, \label{eq:utility_function}
\end{align}
where we generally assume that $\beta_1 > 0$ and $\beta_2 < 0$.  
The \revtwo{parameter $\beta_1$ captures} a \textit{preference for prestige}; \revtwo{a positive value of $\beta_1$ indicates a tendency} to endorse others with high scores.
The \revtwo{parameter $\beta_2$ captures} a \textit{preference for proximity}; \revtwo{a negative value of $\beta_2$ indicates a tendency} to endorse others with scores relatively similar to their own.
Many other choices of utility functions are possible; we prove a stability theorem for a large class of these functions in\articleonly{SI Appendix}{\Cref{sec:bifurcations_proof}}.

\articleonly{\begin{figure}}{\begin{tuftefigure}}
	\begin{center}
		\articleonly{
		\includegraphics[width=1\linewidth]{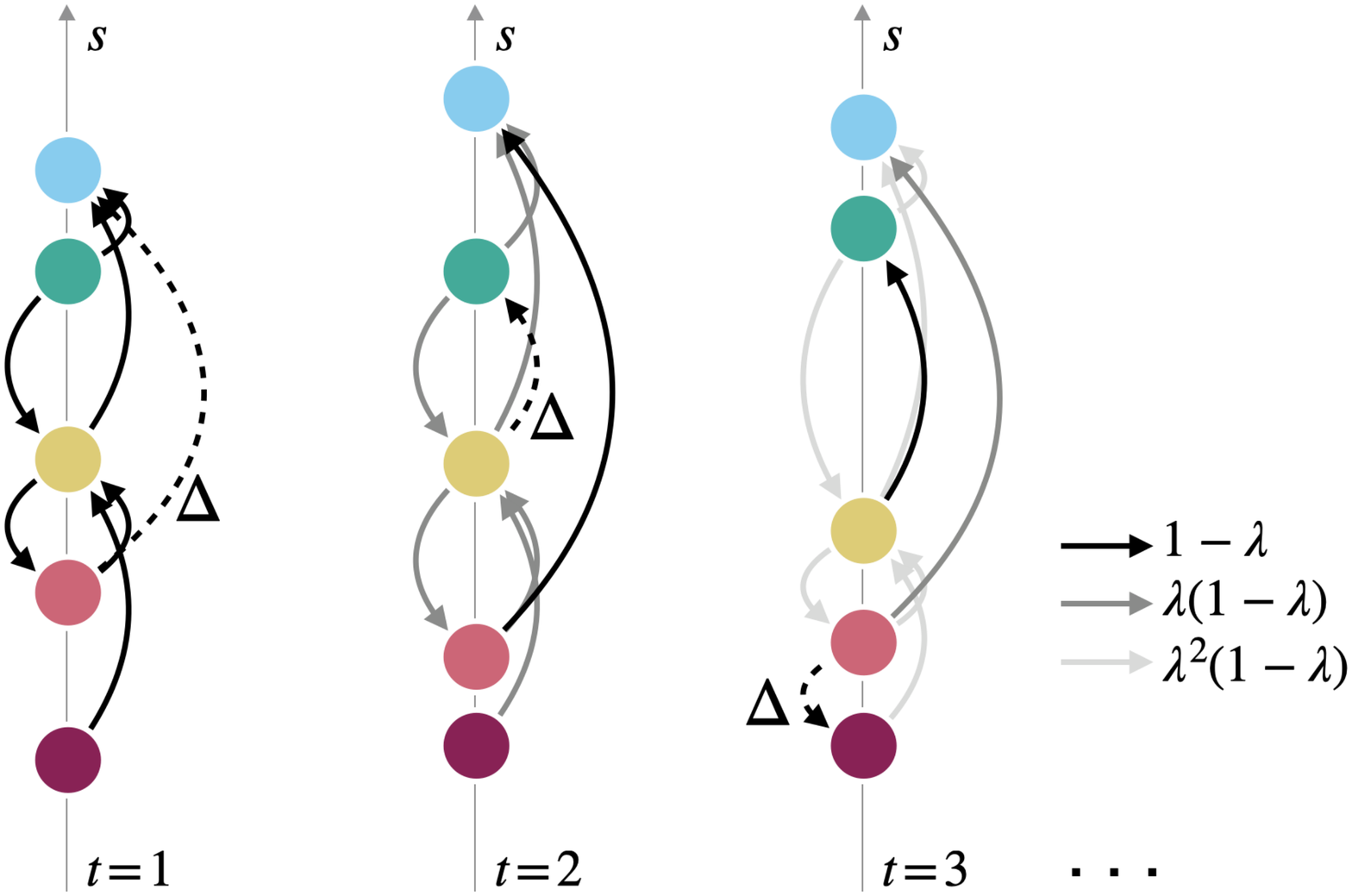}
		}{
		\centering
		\includegraphics[width=0.8\linewidth]{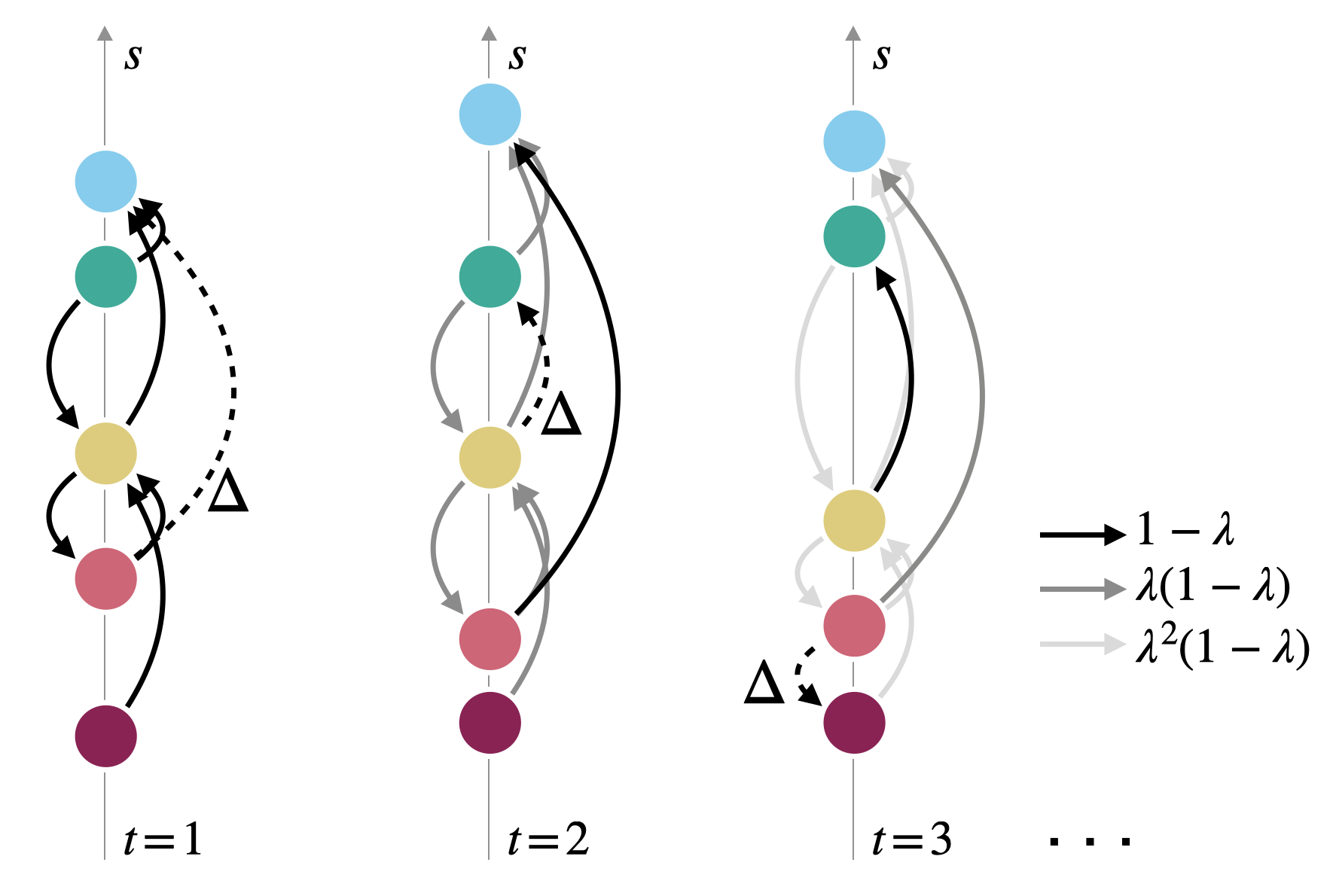}}
		\caption{
			Schematic illustration of our model dynamics. 
			Nodes are initialized at time $t = 1$ with a set of pre-existing endorsements logged in $\mA$ (solid arrows) and the score $\vs = \sigma(\mA)$ is computed (vertical axis). 
			Then, a new edge logged in $\mDelta$ is added (dashed line). 
			In the next time step $t = 2$, old interactions decay by a factor of $\lambda$ (grey arrows). 
			The new endorsement and decay of previous endorsements lead to an updated score function, which then informs the next time step.
		}
		\label{fig:schematic}
	\end{center}
\articleonly{\end{figure}}{\end{tuftefigure}}

In the random utility model, node $i$ observes all possible utilities subject to \rev{noise. 
Traditionally, this noise is chosen to be Gumbel-distributed, in which case the probability that endorsing $j$ yields the greatest utility is given by the multinomial logit} \cite{train2009discrete}
\begin{equation}
	p_{ij}\left(\vs\right) = 
	\dfrac{e^{u_{ij}\left(\vs\right)}}{\sum_{j = 1}^n e^{u_{ij}\left(\vs\right)}}\;. \label{eq:utility}
\end{equation}
\rev{We collect $m\in\mathbb{N}$ endorsements in an update matrix $\mDelta$, where the entry $\mDelta_{ij}$ gives the number of times that $i$ endorses $j$ in the time step.} 
More complex random utility models can lead to more realistic structures in  networks with a growing number of nodes \cite{gupta2020mixed}; we do not pursue these complications here because our model does not focus on network growth, and because these complications obstruct analytical insight. 

\revtwo{
\Cref{eq:utility} can also be derived from an alternative model in which node $i$ makes a randomized choice among $n$ nodes to endorse. In this model, the option to endorse $j$ is assigned a deterministically-observed weight proportional to $e^{u_{ij}(\vs)}$. 
In this case, $\beta_1$ and $\beta_2$ signify inverse temperatures that tune the degree of randomness in this choice, with lower values corresponding to greater randomness.
Although this alternative model---in which node $i$ makes a noisy choice between deterministically-observed utilities---and the random utility model---in which node $i$ makes a deterministic choice between noisily-observed utilities---are mathematically equivalent, the two formulations can lead to different interpretations of system behavior. 
In the case of institutional faculty hiring discussed below (see Hierarchies in Data), the random utility model assumes that a hiring committee makes imperfect observations of the utilities of the institutions from which they could hire, and then deterministically chooses the highest of these imperfectly-observed qualities. 
In contrast, the alternative framework assumes that the committee makes a perfect observation of the utilities, but then chooses among them with some degree of randomness, which may reflect dissension on the hiring committee, search-specific priorities, among other factors. 
} 

\articleonly{\begin{figure}}{\begin{tuftefigure}}
	\articleonly{
		\centering
		\includegraphics[width=0.48\textwidth]{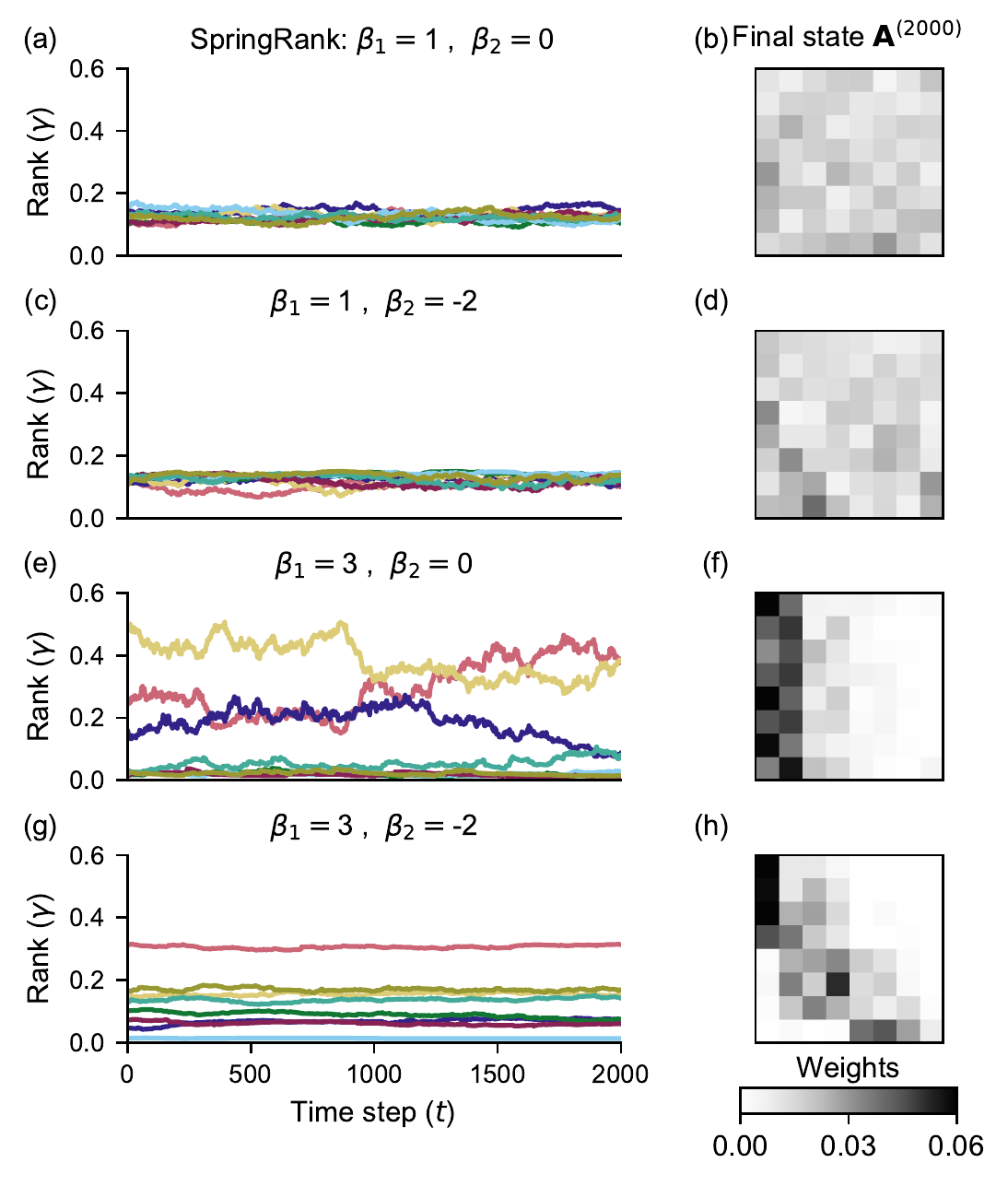}
	}{
		\centering\includegraphics[width=\textwidth]{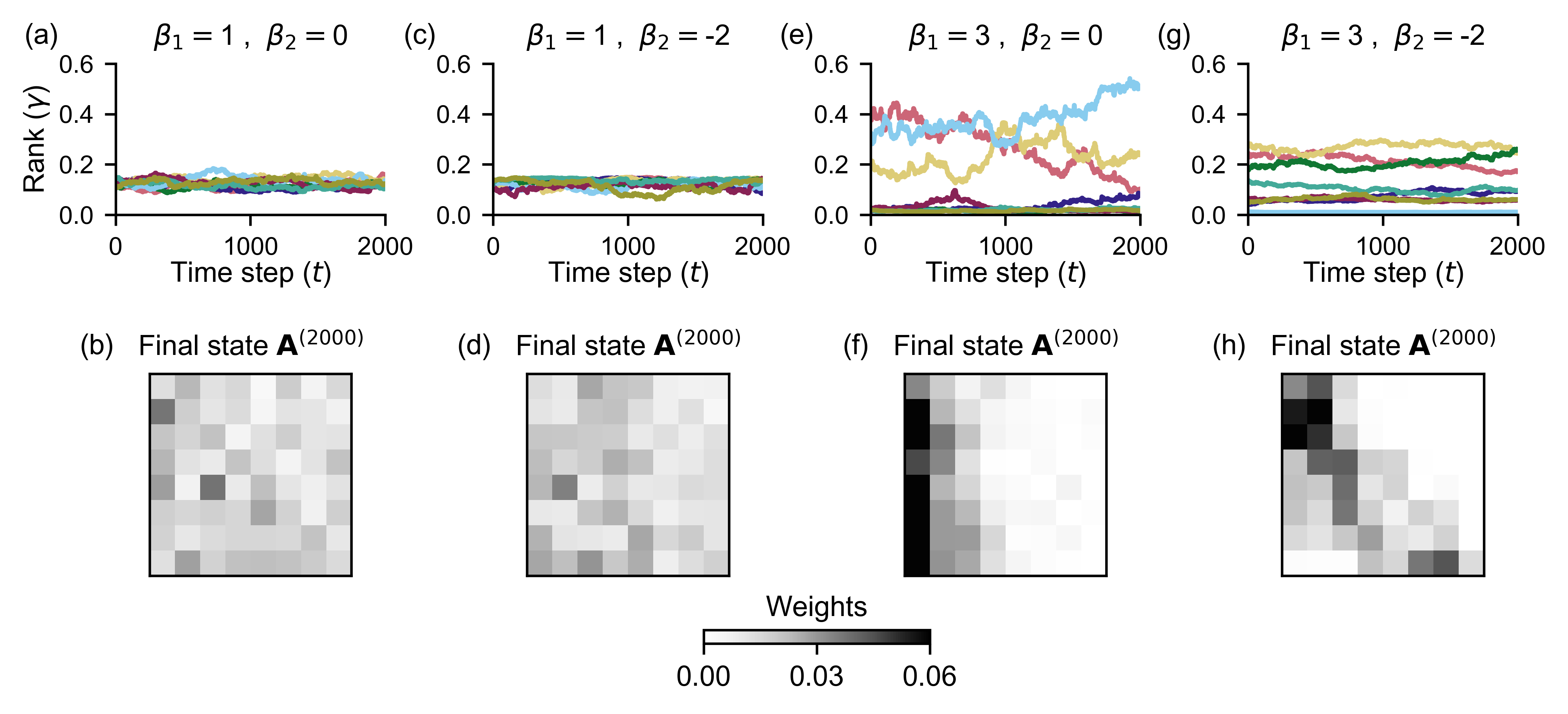}
	}
		\caption{Representative dynamics of \revtwo{the proposed} model. Each column shows a population of $n = 8$ nodes simulated for 2,000 time steps using the SpringRank score function with \rev{$m\,{=}\,1$ update per time step}, varying the preference parameters $\beta_1$ and $\beta_2$. Panels (a),~(c),~(e),~and~(g) show the simulated rank vector $\vgamma$ over time; different colors track the ranks of different nodes. Panels (b),~(d),~(f),~and~(h) show the adjacency matrix $\mA$ at time step $t\,{=}\,2000$ for the corresponding parameter combinations. 
		\rev{See SI Appendix, Fig.~S1 for additional examples with SpringRank; SI Appendix, Fig.~S2 for examples PageRank; and SI Appendix, Fig.~S3 for examples with Root-Degree. See SI Appendix, Fig.~S4 for the dependence of rank variance on $\beta_1$ and $\beta_2$ jointly.} Parameters: $\lambda\,{=}\,0.995$, \rev{$\alpha_s\,{=}\,10^{-8}$}.
	}
	\label{fig:traces}
\articleonly{\end{figure}}{\end{tuftefigure}}

\Cref{eq:dynamics,eq:utility} capture key features of our model. 
First, the dynamics in \cref{eq:dynamics} imply that past interactions decay geometrically at rate $\lambda$. 
This global, gradual decay contrasts with another rank-based relinking model in which single edges fully disappear within each time step \cite{konig2011network}. 
Second, \cref{eq:utility} implies that the likelihood of a node being endorsed at a given time step depends only on the distribution of previous endorsements and \rev{not on intrinsic strength or desirability}.
Those who receive more endorsements and therefore obtain higher scores are more likely to be endorsed in the future---a mechanistic instantiation of winner effects via social reinforcement. 

\Cref{fig:schematic} \revtwo{schematically illustrates model} dynamics \rev{with $m=1$ endorsement per time step.} 
At time $t = 1$, the model is initialized with a small number of endorsements logged in $\mA$. 
\revtwo{The score function takes $\mA$ as an input and outputs the score vector $\vs$, which in turn determines a new interaction according to \cref{eq:utility}}. Logged in $\mDelta$,  \revtwo{this new interaction is} weighted by $1-\lambda$ \revtwo{and added to the previous endorsements, which are} discounted by $\lambda$.  
This process repeats \revtwo{over time} with new endorsements gradually replacing old ones in the system's memory, \revtwo{sequentially updating} the score vector $\vs$. 
\Cref{fig:schematic} also depicts in stylized fashion the operation of both a winner effect ($\beta_1 > 0$), in which endorsements tend to flow in the direction of increasing score, and a proximity effect ($\beta_2 < 0$), in which endorsements tend to flow between nodes of similar scores. 
The net effect is that most endorsements are ``short hops'' up the hierarchy. As we will discuss, this is a common pattern in empirical data.

Despite its simplicity, the model displays a wide range of behaviors. 
To observe them, we define a \textit{rank vector} $\vgamma$, whose $j$th entry $\gamma_j =  n^{-1} \sum_{i} p_{ij}$ gives the likelihood that a new endorsement flows to $j$.
We say that the system state is \emph{egalitarian} when all ranks $\gamma_j$ are equal and \emph{hierarchical} otherwise. 
\Cref{fig:traces} illustrates representative behaviors when the SpringRank score is used. 
When $\beta_1$ is relatively small, winner effects are overtaken by noise, and the system settles into an approximately egalitarian state (\cref{fig:traces}a,b). 
When $\beta_1$ is relatively large, persistent hierarchies emerge  (\cref{fig:traces}{c-f}). 
Moreover, the distribution and stability of ranks depend on the strength of proximity effects, modeled by the quadratic term in the utilities. 
For $\beta_2 = 0$ (no proximity preference), a single node garners more than half of endorsements in a hierarchy with significant fluctuations (\cref{fig:traces}c,d). 
Adding a proximity preference leads to a marginally more equitable hierarchy with ranks that are nearly constant in time (\cref{fig:traces}e,f).

\section*{The Long-Memory Limit}
\label{sec:bif}

The behavior observed in \cref{fig:traces} suggests the presence of qualitatively distinct regimes depending on \revtwo{prestige preference} $\beta_1$. 
For small $\beta_1$ (\cref{fig:traces}a), the winner effect is weak and approximate egalitarianism prevails. 
For larger $\beta_1$, a stronger winner effect enforces a stable hierarchy. 
We characterize the boundary between these regimes analytically in the \emph{long-memory limit} $\lambda \to 1$ by defining a function $\vf$ which is analogous to a deterministic time-derivative for the dynamics of our discrete-time stochastic process.
Let
\begin{align}
    \vf(\vs,\mA) = \lim_{\lambda \rightarrow 1} \frac{\E[\sigma(\lambda \mA + (1-\lambda)\mDelta)] - \vs}{1-\lambda}\;, \label{eq:limit_f}
\end{align}
where the expectation is taken with respect to $\mDelta$. If $\vf(\vs, \mA) = 0$ for all $\mA$, the score vector $\vs$ is a fixed point of the model dynamics in expectation. 
Our choices of Root-Degree, PageRank, and SpringRank score functions admit closed-form expressions for $\vf$, allowing us to analytically derive the conditions for the stability of egalitarianism in the limit of long memory. 

\articleonly{\begin{figure}}{\begin{tuftefigure}}
    \articleonly{
        \includegraphics[width=0.49\textwidth]{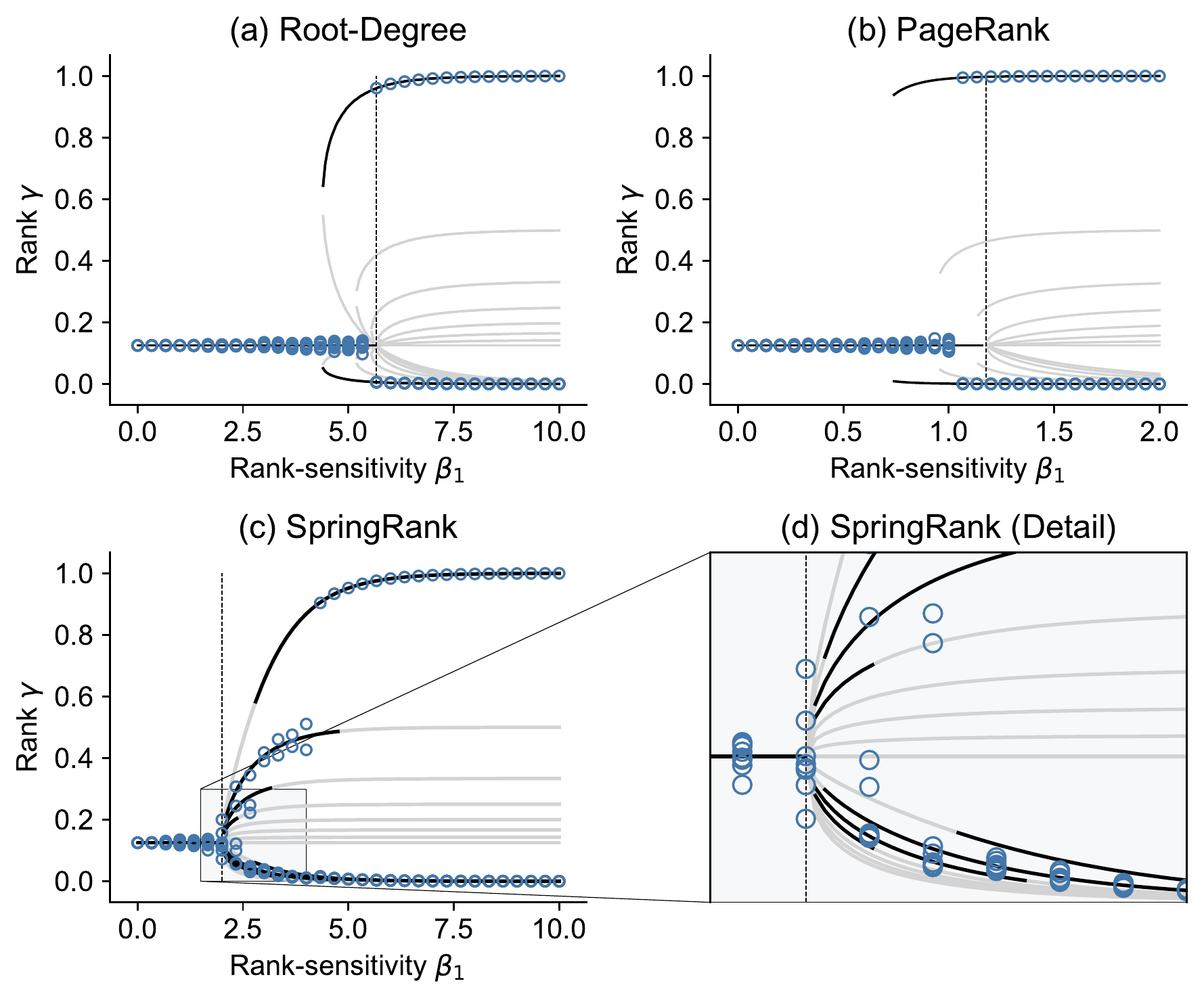}
    }{
        \centering
        \includegraphics[width=.8\linewidth]{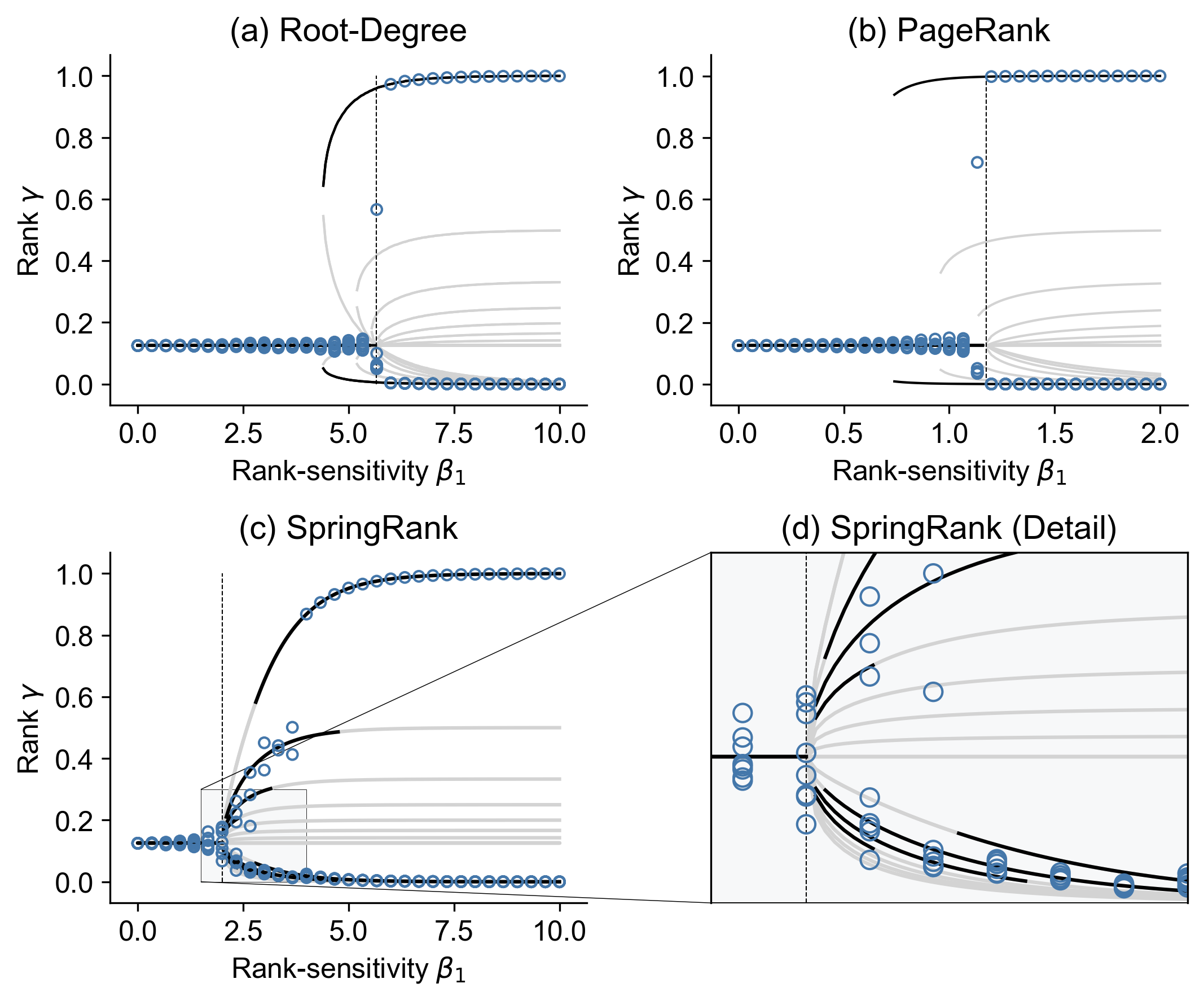}
    }
    \caption{
        Bifurcations in models with Root-Degree, PageRank, and SpringRank score functions with $\beta_2 = 0$ \rev{and $m = 1$ update per time step}.  
        Points give the value of the rank vector $\vgamma$ averaged over the final 500 time-steps of a $5\times 10^{4}$-step simulation with $n = 8$ nodes, memory parameter $\lambda = 0.9995$, and varying $\beta_1$ specified by the horizontal axis. 
        Solid curves show stationary points of the long-memory dynamics obtained by numerically solving the equation $\vf(\vs,\mA) = \vzero$, subject to the restriction that nodes separate into two groups with identical ranks in each. 
        Black curves are linearly stable, while grey curves are unstable. 
        Stability was determined by studying the spectrum of the Jacobian matrix of $\vf$. 
        Vertical lines give the critical value $\beta^c_1$ at which the egalitarian solution becomes linearly unstable according to \Cref{thm:bifurcations}. \rev{Parameters: $\alpha_p = 0.85$, $\alpha_s = 10^{-8}$.}
    }
    \label{fig:bifurcations}
    \articleonly{\end{figure}}{\vspace{-0.15in}\end{tuftefigure}}

\begin{thm} \label{thm:bifurcations}
    For each of the Root-Degree, PageRank, and SpringRank score functions, $\vf$ has a unique egalitarian root. 
    This root is linearly stable if and only if $\beta_1 < \beta_1^c$, where 
    \begin{align*}
        \beta^c_1 = \begin{cases}
            2\sqrt{\dfrac{n}{m}} &\quad \text{Root-Degree},\\
            1/\alpha_p &\quad \text{PageRank}, \\
            2 + \alpha_s \dfrac{n}{m} &\quad \text{SpringRank}.\\ 
        \end{cases}
    \end{align*}
\end{thm}
\noindent In\articleonly{SI Appendix}{\Cref{sec:bifurcations_proof}}, we prove \Cref{thm:bifurcations}, \rev{as well as a generalization to arbitrary smooth utility functions.} 
In each case, the proof of uniqueness exploits the algebraic structure of the score function, and the critical value $\beta^c_1$ is obtained via the linearization of $\vf$ about the egalitarian state. 
Interestingly, only $\beta_1$ plays a role in the stability of the egalitarian root. 
\rev{While \revtwo{proximity preference} $\beta_2$ does not determine where the hierarchical regime begins, it does influence the structure of and the transient dynamics toward nonegalitarian equilibria (\cref{fig:traces}(e),(g)).}

\Cref{fig:bifurcations} illustrates the destabilization of egalitarianism predicted by \Cref{thm:bifurcations} in the case of $n = 8$ nodes.
Although not required by \Cref{thm:bifurcations}, we fix $\beta_2 = 0$ for simplicity.
Curves show fixed points of the model dynamics in the long-memory limit. 
We show only fixed points in which nodes separate into two groups, each of which have identical rank. 
For $\beta_1 < \beta^c_1$, the egalitarian regime is stable and the long-run state deviates from egalitarianism only slightly. 
For $\beta_1 > \beta^c_1$, in contrast, the long-run state switches to an inegalitarian, stable fixed point. 

In the Root-Degree and PageRank models, there is a single stable inegalitarian equilibrium with one node absorbing nearly all endorsements (\cref{fig:bifurcations}a,b). 
Interestingly, there is a bistable regime in which both egalitarian and inegalitarian states are attracting. Whether the system converges to one or the other depends on initial conditions. 
The SpringRank model displays qualitatively distinct behavior (\cref{fig:bifurcations}c,d). 
Past $\beta^c_1$, we observe staggered multistable regimes. As $\beta_1$ increases, equilibria with multiple elite (i.e., highly ranked) nodes become sequentially unstable until eventually only a single elite node remains. 
The long-term behavior of the system again depends on initial conditions, but now there are many more possible stable states. 
This behavior would seem to make the SpringRank score function especially appropriate for modeling empirical systems with multiple distinct hierarchical regimes and sensitivity to initial conditions, an intuition which we confirm empirically in the following section.

\section*{Hierarchies in Data} \label{sec:inference}

In addition to being amenable to analytical treatment, our model has
a tractable likelihood function, described in\articleonly{SI Appendix}{\Cref{sec:inference}}.
This allows us to study hierarchical structures in empirical data using principled statistical inference. 
The likelihood function not only supports maximum-likelihood parameter estimates of $ \lambda$, $\beta_1$, and $\beta_2$ but also enables direct comparisons of different score functions in a statistically rigorous framework: score functions with higher likelihoods provide more predictive low-dimensional summaries of observed interactions.  
This in turn allows us to explore the relative value of competing mechanistic explanations of observed data. 

Several mathematical features of the model facilitate the exploration of real data. First, the predictive distribution \cref{eq:utility} is in the linear exponential family, making the estimation of $\vbeta$ a convex optimization with a unique solution. Second, the estimation problem in $\hat{\lambda}$ is in general nonconvex, but can be tractably solved via first-order optimization methods with multiple starting points. Finally, while model likelihoods evaluated on training data may in principle be inflated due to overfitting, our model uses only three parameters to fit hundreds or thousands of observations, suggesting that overfitting is not a major concern. 

We conducted a comparative study of model behavior on four data sets: an academic exchange network in math, two networks of parakeet interactions, and a network of friendships among members of a fraternity.
The Math~PhD~Exchange data set is extracted from The Mathematics Genealogy Project \cite{mathgene,taylor2017eigenvector}.
Nodes are universities. 
An interaction $i\rightarrow j$ at time $t$ occurs when a mathematician who received their degree from university $j$ at time $t$ supervises one or more PhD theses at university $i$.
This event is a proxy for university $i$ hiring a graduate from university $j$ at a time near $t$. 
We view this as an endorsement by $j$ that graduates of $i$ are of high quality~\cite{clauset2015hierarchy}.
We restricted our analysis to the activity of the 70 institutions that placed the most graduates between 1960 and 2000. Doing so helped to avoid singularities produced by institutions with no placements early in the time period and to minimize temporal boundary effects associated with the beginning and end of data collection.

The two Parakeet data sets \cite{hobson2015social} record aggression events in two distinct groups of birds studied over four observation quarters (weeks). 
An interaction $i\rightarrow j$ at~time $t$ occurs when parakeet $i$ loses a fight to parakeet $j$ in period~$t$.
Since there are just four observation periods, estimates of the memory parameter $\lambda$ should be approached with caution. 

Lastly, the Newcomb Fraternity data set was collected by the authors of 
\citet{nordlie1958longitudinal} and \citet{newcomb1916acquaintance}
\rev{and accessed via the KONECT network database \cite{kunegis2013konect}.}
The data set documents friendships among members of a fraternity at the University of Michigan. 
\rev{
Each week during a fall semester, excluding a week for fall break, 
each of 17 cohabiting brothers ranked every other brother according to friendship preference, with ranks 1 and 16 referring to that brother's most and least preferred peers, respectively. 
}
An endorsement $i\rightarrow j$ is logged when brother $i$ ranks $j$ among his top $k = 5$ peers (small changes to $k$ did not significantly alter the results). 
While friendship is often viewed as a symmetric relationship, expressed friendship preferences may be asymmetric \cite{carley1996cognitive}.  

\setcounter{figure}{0}
\renewcommand{\figurename}{Table} 
\articleonly{\begin{figure}}{\begin{tuftefigure}}
    \begin{center}

\begin{tabular}{l l r r r}
                        &                 &  Root-Degree     &   PageRank      &  SpringRank    \\ 
    \toprule
    Math PhD            & $\hat{\lambda}$ &  0.87 (0.01)     &   0.96 (0.01)   &  0.91 (0.01)   \\ 	
    Exchange            & $\hat{\beta}_1$ &  1.28 (0.02)     &   0.74 (0.01)   &  2.99 (0.04)   \\ 
    (\emph{N} = 6,019)  & $\hat{\beta}_2$ & -0.18 (0.01)     &  -0.07 (0.00)   & -1.12 (0.04)   \\ 
    \cmidrule{2-5}
                        & $\mathcal{L}$   & \textbf{-14,379} & -15,001         & -14,927        \\
    \midrule
    Parakeets (G1)      & $\hat{\lambda}$ &  0.97 (0.08)     &  0.59 (0.08)    &  0.67 (0.14)   \\ 	
    (\emph{N} = 838)    & $\hat{\beta}_1$ &  0.84 (0.05)     &  1.82 (0.08)    &  3.03 (0.16)   \\ 
                        & $\hat{\beta}_2$ & -0.12 (0.01)     & -0.50 (0.03)    & -1.74 (0.12)   \\ 
    \cmidrule{2-5}
                        & $\mathcal{L}$   & -1,106           & -1,053          & \textbf{-964}  \\
    \midrule 
    Parakeets (G2)      & $\hat{\lambda}$ &  0.42 (0.07)     &  0.13 (0.03)    &  0.40 (0.06)   \\ 	
    (\emph{N} = 961)    & $\hat{\beta}_1$ &  0.62 (0.03)     &  0.82 (0.04)    &  2.86 (0.14)   \\ 
                        & $\hat{\beta}_2$ & -0.06 (0.01)     & -0.12 (0.01)    & -1.46 (0.12)   \\ 
    \cmidrule{2-5}
                        & $\mathcal{L}$   & -975             & -1029           & \textbf{-924}  \\
    \midrule 
    Newcomb             & $\hat{\lambda}$ &  0.56 (0.13)     &  0.81 (0.19)    &  0.71 (0.14)   \\ 	
    Fraternity          & $\hat{\beta}_1$ &  0.95 (0.05)     &  1.21 (0.07)    &  2.33 (0.14)   \\ 
    (\emph{N} = 1,428)  & $\hat{\beta}_2$ & -0.08 (0.03)     & -0.25 (0.05)    & -0.86 (0.16)   \\ 
    \cmidrule{2-5}
                        & $\mathcal{L}$   & -1,850           & -1,865          & \textbf{-1,841}\\
    \midrule
\end{tabular}
    \end{center}
    \caption{
    Parameter estimates and likelihood scores using each of three score functions for the four data sets \rev{described in the main text}. 
    \rev{Parenthetical values are} standard errors for each parameter estimate.  
    For each data set, the largest log-likelihood $\cL$ is indicated in \textbf{bold}. 
    All parameter estimates are statistically distinct from zero at 95\% confidence. 
    $N$ gives the total number of interactions in the data. 
    \rev{See\articleonly{SI Appendix, Fig.~S5}{\cref{fig:trace-inferred-params}} for simulated trajectories with the inferred parameters.}
    }\label{tb:comparative}
\articleonly{\end{figure}}{\end{tuftefigure}}
\renewcommand{\figurename}{Fig.}

We studied these data using the Root-Degree, PageRank, and SpringRank score functions. 
Table \ref{tb:comparative} summarizes our results, including parameter estimates; standard errors (obtained by inverting the numerically-calculated Fisher information matrix); and optimized log-likelihoods for each combination of score and data set. 
Several features stand out. 
In all four data sets and across all three score functions, we find $\hat{\beta}_1 > 0$ and $\hat{\beta}_2 < 0$. 
This suggests a persistent pattern in time-dependent hierarchies: while endorsements do flow upward ($\hat{\beta}_1 > 0$), nodes are more likely to endorse those close to them in rank ($\hat{\beta}_2 < 0$).
Endorsements tend to flow a few rungs up the ladder---not directly to the top.  
The reasons for this pattern likely vary across data set. 
In the Math PhD Exchange, this may indicate that low-ranked schools struggle to recruit graduates of \revtwo{high-ranked ones} due to a limited supply of elite candidates. 
In parakeet populations, proximal aggression may facilitate inference of dominance hierarchies through transitive inference \cite{hobson2015social}. 
In Newcomb's Fraternity, we postulate that implicit social norms may encourage friendships between those of similar standing. Similar results have been reported in static social network data among adolescents~\cite{ball2013friendship}. 
Thus, while we do not attribute this pattern in the parameter estimates to a ``universal'' mechanism, we suggest its persistence as an interesting observation worthy of future study. 

Because different score functions capture distinct qualitative features of the data, quantitative comparisons yield insights into the generating mechanisms at work.
In general, parameters from models using differing score functions should not be directly compared, since these parameters are sensitive to the scale of the score vector. 
However, we can compare models on the basis of their likelihoods. 
In the Math PhD Exchange, the Root-Degree model was strongly favored over either SpringRank or PageRank.
In the context of this data set, the Root-Degree score is a measure of faculty production: a school that places more candidates has a higher score, regardless of the prestige of the institutions at which the candidates land. 
The strong fit from the Root-Degree score is consistent with previous findings that raw faculty production plays a major role in structuring the hierarchy of academic hiring within computer science, business, and history \cite{clauset2015hierarchy}. 
As \citet{clauset2015hierarchy} note, transitive prestige also plays an important role. 
It would be of significant interest to extend our study to include multiple score functions, enabling an inferential analysis of the relative roles of production and transitive prestige. 

In contrast, the SpringRank score was favored by large margins in both Parakeet data sets and by a smaller margin in the Fraternity data set, suggesting that transitive prestige plays a more prominent role. 
Among Parakeets, it may matter not only how many confrontations one wins but also against whom, with victories over high-ranking birds counting more towards one's own prestige. 
This finding is consistent with those of \citet{hobson2015social}, who found, using different methodology, that parakeet behavior suggests the ability to draw sophisticated, transitive inferences about location in the hierarchy. 
Similarly, in Newcomb's Fraternity, friendships with highly-ranked brothers may confer greater prestige than those with lower-ranked ones. 

In addition to the likelihoods, we can also compare the memory estimate $\hat{\lambda}$ across models and data sets. 
Since the model assumes that the impact of past endorsements decays at rate $\lambda$, the quantity $t_{1/2} = -\log(2)/\log(\hat{\lambda})$ represents the half-life of system information according to the inferred dynamics, in units of observation periods. 
In the Math PhD data, the favored Root-Degree score gave a half-life of $t_{1/2} \approx 5$ years.
In the Parakeets data, the half-life estimated under SpringRank is $t_{1/2} \approx 1.7$ weeks for the first group and $t_{1/2} \approx 0.8$ weeks for the second. 
The small number of observation periods implies that these estimates should be approached with caution. 
Finally, in the Newcomb Fraternity data, the SpringRank half-life was $t_{1/2} \approx 2$ weeks.
This suggests that the friendships in this data set evolved on timescales much shorter than the full semester.
This likely reflects the fact that the brothers did not know each other prior to data collection, requiring them to form their social relationships from scratch.
\rev{
    An important caveat in interpreting these estimated half-lives is that the \emph{indirect} influence of an interaction may extend far beyond its \textit{direct} influence. 
    In the Math PhD data, for instance, while the half-life indicates that only a quarter of hiring events will be directly ``remembered'' in the system after a decade, those events will have influenced ten cycles of hiring, which may further reinforce the patterns established by the earlier events. 
    }

\renewcommand{\figurename}{Table} 
\articleonly{\begin{figure}}{\begin{tuftefigure}}
        \begin{center}
            \rev{

\begin{tabular}{l l l l l}
                        &                 &  Root-Degree     &   PageRank      &  SpringRank    \\ 
    \toprule
    Math PhD            & $\beta_1^c$ &  \textbf{1.36}     &   1.18   &  2.00   \\ 	
    Exchange            & $\hat{\beta}_1$   &  \textbf{1.28}$_*$ (0.02)     &   0.74$_*$ (0.01)   &  2.99* (0.04)   \\ 
    \midrule
    Parakeets (G1)      & $\beta_1^c$ &  0.55      &  1.18     &  \textbf{2.00}    \\ 	
                        & $\hat{\beta}_1$   &  0.84$^*$ (0.05)     &  1.82$^*$ (0.08)    &  \textbf{3.03}$^*$ (0.16)   \\ 
    \midrule 
    Parakeets (G2)        & $\beta_1^c$ &  0.49      &  1.18     &  \textbf{2.00}    \\ 	
                           & $\hat{\beta}_1$ &  0.62$^*$ (0.03)     &  0.82$_*$ (0.04)    &  \textbf{2.86}$^*$ (0.14)   \\ 
    \midrule 
    Newcomb             & $\beta_1^c$ &  0.89      &  1.18    &  \textbf{2.00}    \\ 	
    Fraternity          & $\hat{\beta}_1$ &  0.95 (0.05)     &  1.21 (0.07)    &  \textbf{2.33}$^*$ (0.14)   \\ 
    \midrule
\end{tabular}
            }
        \end{center}
    \caption{
        \rev{
            Estimates of $\beta_1$ (identical to those in \Cref{tb:comparative}) compared to the mean critical value $\beta_1^c$ for each system. 
            $\beta_1^c$ is calculated as in \Cref{thm:bifurcations}, using as $m$ the mean number of interactions per time-step in the observed data. 
            As in Table \ref{tb:comparative}, the parameters corresponding to the highest log-likelihood are shown in \textbf{bold}. 
            Estimates shown with an upper asterisk$^*$ exceed the approximate critical value by two standard errors, while estimates shown with a lower asterisk$_*$ are smaller than the approximate critical value by two standard errors. 
        }
    	\rev{
    		See\articleonly{SI Appendix, Fig.~S5}{\cref{fig:trace-inferred-params}} for simulated trajectories using the inferred parameters.
    	}
    }
    \label{tb:criticality}
\articleonly{\end{figure}}{\end{tuftefigure}}
\renewcommand{\figurename}{Fig.}

\rev{
    As described in \Cref{thm:bifurcations}, in the long-memory limit, our model has distinct egalitarian and hierarchical regimes, separated by a critical value $\beta_1^c$. 
    The model's estimate of $\beta_1$ allows us to roughly locate empirical systems within these regimes.
    There are two necessary points of caution. 
    First, when the estimate $\hat{\lambda}$ is far from the idealized long-memory limit, hierarchical and egalitarian regimes may not be sharply distinguished. 
    Second, in the Math PhD and Parakeet data, the number of updates $m$ varies between time steps. Here, a reasonable approximation is to use the average number of updates $\bar{m}$ per time step. 
    Using this average and \Cref{thm:bifurcations}, we computed an approximate long-memory critical value $\beta_1^c$ for each empirical system. 
}

\revtwo{
    Comparing the data-derived preference estimates $\hat{\beta}_1$ to the approximate critical values $\beta_1^c$ reveals that all four empirical systems are in or near the hierarchical regime (Table \ref{tb:criticality}).
}
\rev{
    The Root-Degree estimates of $\beta_1$ tend to be very close to the approximate critical point.
    For the Math PhD data, in which Root-Degree is the preferred model, the estimate  \revtwo{of $\beta_1$} is slightly but statistically-significantly below the critical value.
    In each of the other three data sets, the estimate is slightly above the critical value, \revtwo{and} significantly so in the two Parakeet groups. 
    \revtwo{Given} the presence of a bistable regime in the Root-Degree model (\cref{fig:bifurcations}(a)), the estimate of $\beta_1$ for the Math PhD data is consistent with persistent hierarchy despite the fact that the estimate falls slightly below the critical threshold.
    Indeed, simulations with the inferred parameters produce persistent hierarchical structure similar to that observed in the data (SI Appendix, Fig.~S5).
    The PageRank estimates behave similarly \revtwo{to Root-Degree}, although the finding in Parakeets (G2) is reversed. 
    The presence of a bistable regime in the PageRank model (\cref{fig:bifurcations}(b)) indicates that these findings are \revtwo{consistent} with persistent hierarchy in any of these data sets (see SI Appendix, Fig.~S5 for simulated dynamics). 
    Finally, in the SpringRank model, which obtains the highest likelihood for both Parakeet data sets and the Newcomb Fraternity, the estimated values of $\beta_1$ significantly exceed the estimated critical values, and tend to lie in or near the range $[2,3]$.  
    In summary, all three models suggest that the system \revtwo{corresponding to each data set} is in or near the regime of self-reinforcing hierarchy. 
}

\setcounter{figure}{3}
\articleonly{\begin{figure}}{\begin{tuftefigure}}
\articleonly{
        \includegraphics[width=.49\textwidth]{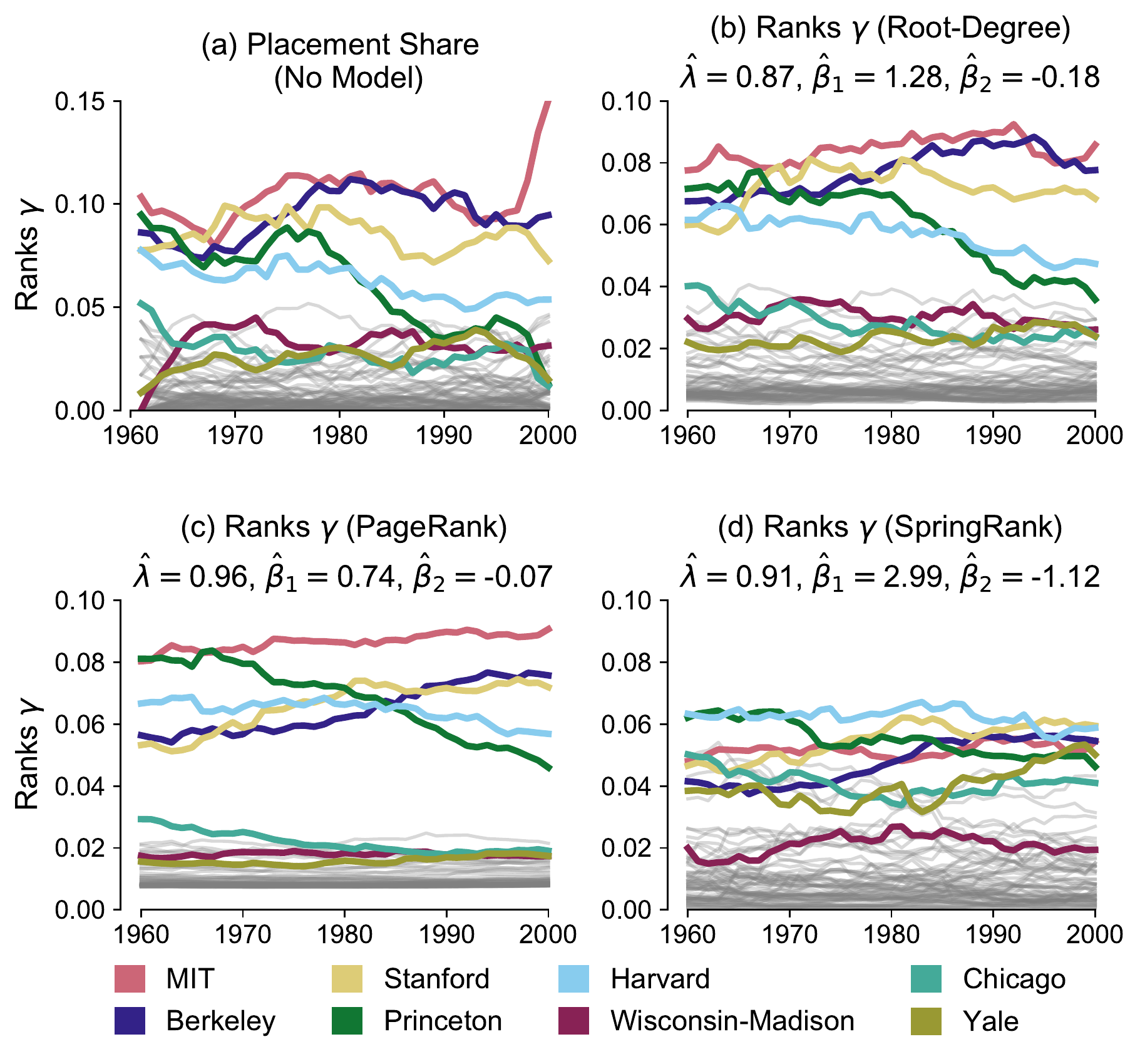}
    }{
        \centering
        \includegraphics[width=.8\textwidth]{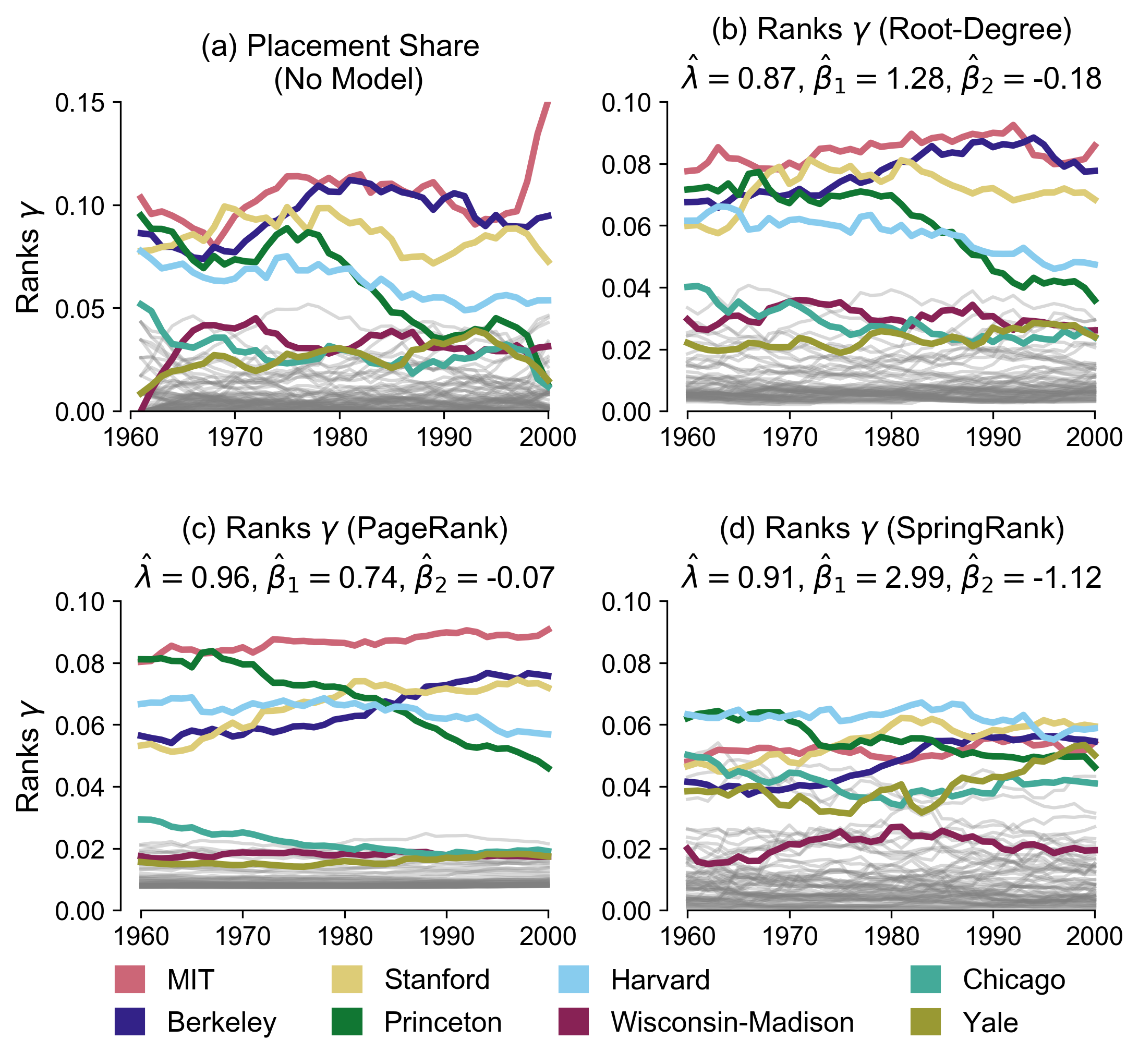}
    }
    \caption{
        Visualization of evolving ranking functions in the Math PhD Exchange.
        (a): Fraction of all placements (number of graduates hired) from each school, shown as a moving average with bin-width 8 years for visualization purposes.
        (b): Inferred rank vector $\vgamma$ as a function of time using the Root-Degree score function. 
        (c-d): As in (b), with PageRank and SpringRank score functions, respectively. 
        Parameters for panels (b-d) are shown in the first section of Table~\ref{tb:comparative}. 
    }
    \label{fig:math_phd_trajectories}
\articleonly{\end{figure}}{\end{tuftefigure}}

Our model also assigns intepretable, time-dependent ranks to empirical data (\cref{fig:math_phd_trajectories}).
\revtwo{For the Math PhD Exchange network, for example}, the raw placement share (\cref{fig:math_phd_trajectories}a) and Root-Degree model (\cref{fig:math_phd_trajectories}b) show strong qualitative agreement, with institutions that place the most candidates occupying higher ranks. 
Due to the relatively large estimates $\hat{\lambda}$, both the Root-Degree  and PageRank models (\cref{fig:traces}b-c) produce smoother rank trajectories than the purely-descriptive placement share with 8-year rolling average. 
In contrast, the SpringRank score generates qualitatively different trajectories that are less sensitive to raw volume (\cref{fig:math_phd_trajectories}d). 
For instance, SpringRank places Harvard at the top over most of the time period, while the other scores prefer MIT. This difference reflects SpringRank's sensitivity to \emph{where} Harvard's graduates were placed, a consideration which Root-Degree entirely ignores. 
Similarly, SpringRank places Chicago and Yale noticeably higher than Wisconsin-Madison, despite all three having similar numbers of placements. 

\section*{Discussion}
\label{sec:discussion}

We have proposed a simple and flexible model of persistent hierarchy as an emergent feature of networked endorsements with feedback. 
When the \revtwo{preference} for high status exceeds a critical value, egalitarian states destabilize and hierarchies emerge. 
The location of this transition depends on the structure of the score function and of the node's preferences. 
Our findings emphasize that winner effects do not require internal, \rev{rank}-enhancing feedback mechanisms. 
Social reinforcement through prestige preference is sufficient to generate social hierarchies. 

Crucially, our model has a tractable likelihood function, supporting principled statistical inference of parameters---for both preferences and memory strength---from empirical data. In the four data sets analyzed, we found that links are typically formed in alignment with the hierarchy ($\hat \beta_1 > 0$) but that they are preferentially created to other nodes with similar ranks ($\hat \beta_2 < 0$). The likelihood also opens the door to model selection to determine relevant score functions. We found that networked ranking methods that capture transferable prestige are preferred over non-networked methods in some but not all systems. Due to its flexibility, our framework can be applied to additional data sets, score functions, and/or preference models to test the generality of these empirical observations.

There are limitations to our approach. 
First, we specified a fixed parametric form for the utilities with \cref{eq:utility_function} and Gumbel-distributed noise with \cref{eq:utility}. 
Other choices may be more justified in particular applications, ideally informed by domain-specific considerations. 
\revtwo{Importantly}, our inferential framework \revtwo{allows for quantitative evaluation and comparison of these choices. Taking advantage of this, future work could} systematically explore the most appropriate functional forms in systems from diverse scientific domains.
Second, \revtwo{our model assumes that} all nodes use identical preference parameters $\beta_1, \beta_2$ and score vector $\vs$ when computing utilities. 
The latter is an especially strong assumption, since it requires each node to have global knowledge of the endorsement network, or at least of the score vector. 
This is unlikely to be true in real systems, and should be regarded as a modeling device.
Future work, along the lines of \citet{hobson2015social}, could explore the interplay between the cognitive capabilities of individuals represented by nodes and the information available to them in the formation of social hierarchies. 

Our model \revtwo{points to} several other avenues for further work. 
A crucial step would be to extend extant network-based models \cite{konig2011network,konig2014nestedness,krause2013spontaneous} so that their parameters could be statistically learned from data. 
This would enable comparative validation of different modeling frameworks. 
Studies of the relationship between measures of time-dependent centralities \cite{taylor2017eigenvector,taylor2019supracentrality,liao2017ranking} and dynamic models of hierarchy would also be valuable. 
In particular, the theory of time-dependent centralities \revtwo{faces an important methodological issue}: 
different reasonable ranking methods can yield directionally different orderings of nodes when applied to the same data set \cite{mariani2020network}. 
Their performance on external validation tasks, such as the prediction of central nodes in spreading processes \cite{lu2016vital}, may also vary significantly. 
Because the theories of centrality and generative networks have evolved largely separately, \revtwo{evaluating} the suitability of a centrality metric for a given dynamic system can be difficult. 
Our inferential approach offers \revtwo{a candidate validation task to overcome this challenge}:
good \revtwo{centrality metrics} are those which most effectively predict the future evolution of the system. 
This approach enables us to not only compare different score and utility functions in a principled manner but also explore their relative importance in observed networks. 
For instance, one could study the relative influence of degree-based and SpringRank scores by incorporating both into our model and then analyzing their distinct coefficients. \revtwo{Further work in this direction could reveal} how different forms of centrality combine to govern the evolution of interaction networks. 
We anticipate that \revtwo{a} fruitful dialogue between centrality theory and generative models of time-varying networks will deepen our understanding of the feedback mechanism between local interactions and hierarchical structures.

\section{Published Article}

    This article was published in the \href{https://www.pnas.org/content/118/16/e2015188118}{\emph{Proceedings of the National Academy of Sciences}}. 
    We request that all citations refer to the published version. 

\section*{Software}
A repository containing our data, model implementation, and figure generation scripts is available at 
{\url{https://github.com/PhilChodrow/prestige_reinforcement}}.
\articleonly{}{\section*{Acknowledgments}}
\revtwo{We thank the anonymous reviewers for their constructive feedback.}
We \revtwo{also thank} Dakota S. Murray (Indiana University, Bloomington), Kate Wootton (Swedish University of Agricultural Sciences; \rev{University of Colorado Boulder}), VPS Ritwika (University of California, Merced), and Rodrigo Migueles Ram\'irez (McGill University) for extremely helpful discussions during the early phase of this work. 
We are also grateful to the organizers of the Complex Networks Winter Workshop in Quebec City at which this work was conceived, and to the inhabitants of Caf\'e F\'elin Ma Langue Aux Chats, Quebec City, Canada for contributing to an exceptionally productive working environment. 
\rev{MK acknowledges support from the Army Research Office Grant W911NF-18-1-0325.}
PSC acknowledges support from the National Science Foundation Award 1122374. DBL acknowledges support from the National Science Foundation Award SMA 1633791 and the Air Force Office of Scientific Research Award FA9550-19-1-0329.
\section*{Gender Representation in Cited Work}

    Recent work in several fields of science has identified gender bias in citation practices: papers by women and other minoritized groups are systematically under-cited in their fields \cite{mitchell2013gendered,dion2018gendered,caplar2017quantitative, maliniak2013gender, Dworkin2020.01.03.894378}.
    Gender bias can arise through explicit and implicit bias against a person's known gender identity, or against a name commonly used by members of a marginalized gender identity \cite{macnell2015s,paludi,moss2012science}.
    In this work, we proactively sought to include relevant citations by non-male authors. 

    We manually gender-coded the authors in the works cited according to personal acquaintance, instances of pronoun usage online, or first name. 
    This method is limited: names and pronouns may not be indicative of gender; gender may change over time; and manual coding is inherently subjective and subject to error. 
    We focused on the first and last authors because typically, though not always, the former is the leading researcher and the latter the senior author in the disciplines included in our references.
    Of the works cited in the main text, 29\% had a non-male first author and 13\% had a non-male last author. 
    Of those with at least two authors,
    38\% had either a non-male first author or a non-male last author. 

    This statement is modeled on those found in \citet{torres2020and} and \citet{Dworkin2020.01.03.894378}. 
    We join those authors and many others in calling for collective effort to promote equitable practices in science.

\pagebreak
\begin{fullwidth}
    \bibliographystyle{dgleich-bib}
    \bibliography{preprint_refs}
\end{fullwidth}

\pagebreak

\begin{flushright}%
    \textbf{\MakeTextUppercase{\allcapsspacing{Supplementary Information}}} \\ 
    \textit{Emergence of Hierarchy in Networked Endorsement Dynamics}
    \bigskip\par%
\end{flushright}\noindent%

\appendix


\section{Linear Stability} \label{sec:bifurcations_proof}
In this section, we prove a set of linear stability results that generalize Theorem 1 in the main text. 
Our generalizations account for (a) nonlinear features and (b) multiple updates per round.  

Throughout this section, we consider a utility function of the form 
\begin{align}
	u_{ij}(\vs) = \sum_{\ell = 1}^k \beta_\ell\phi_{ij}^{\ell}(\vs) \;,
	\label{eq:utility_general}
\end{align}
where each $\phi^{\ell}:\mathbb{R}^n \mapsto \mathbb{R}^{n\times n}$ is a smooth \emph{feature map}; $\beta_\ell \in \mathbb{R}$ is a \textit{preference parameter} indicating relative importance of the $\ell$th feature; and $\phi_{ij}^{\ell}(\vs)$ is the $ij$th entry of $\phi^{\ell}(\vs)$.  
We collect the parameters $\beta$ in a vector $\vbeta \in \R^k$.
The utility function in\articleonly{~Eq.~(4)}{\cref{eq:utility_function}} from the main text is a special case with linear feature map $\phi^{1}_{ij} (\vs) = s_j$, and quadratic feature map, $\phi^{2}_{ij} (\vs) = (s_i - s_j)^2$. 
We also define the \textit{rate matrix} $\mG = [n^{-1} p_{ij}]$, whose $(i,j)th$ entry gives the probability that, in a given time step, node $i$ chosen uniformly at random endorses node $j$ (see\articleonly{~Eq.~(5)}{\cref{eq:utility}} in the main text for the definition of $p_{ij}$). 

Since we aim to characterize the linear stability of egalitarian fixed points, we will consider the Jacobian of the rank vector $\vgamma$ evaluated at egalitarian fixed points. 
We will therefore evaluate the Jacobian at $\vs_0 = \theta \ve$, where $\theta \in \R$. 
By definition, $\vgamma = n^{-1}\mG^T\ve = n^{-1}\sum_i \vgamma_i$, where $\vgamma_i$ is the $i$th column of $\mG$.  
Differentiating and applying the chain rule, we have 
\begin{align*}
    \dds{\gamma(\vs_0)} &= \sum_{i}\left(\mGamma_i - \vgamma_i\vgamma_i^T\right)\sum_{\ell = 1}^k\beta_\ell \dds{\phi_{i}^{\ell}}\;,
\end{align*}
where $\mGamma_i = \diag \gamma_i$ and  $\phi_{i\cdot}^{\ell}(\vs_0)$ is the $i$th row of the $\ell$th feature map evaluated at $\vs_0$. 
At $\vs_0 = \theta\ve$, $\mG = n^{-1}\mE$. 
It follows that $\vgamma_i = n^{-1}\ve$ and $\mGamma_i = n^{-1}\mI$.
We thus have 
\begin{align}
    \dds{\vgamma(\vs_0)} = n^{-1}(\mI - n^{-1}\mE)\sum_{i = 1}^n\sum_{\ell = 1}^k\beta_\ell \dds{\phi_{i\cdot}^{\ell}(\vs_0)} \triangleq \mM(\vs_0; \vbeta)\;. \label{eq:dgamma}
\end{align} 
We will express our primary results in terms of this matrix. 

When writing proofs involving dynamics, we will typically repress the time-argument of quantities like $\vs$ and $\mA$. 
When time step $t$ is implied, we will use the somewhat informal notation $\delta \vs = s(t+1) - \vs(t)$ and $\delta\mA = \mA(t+1) - \mA(t)$ to denote the increments of these and other quantities in the current time step. 

\subsection{Degree Scores}

\begin{thm}[Stable Egalitarianism with Degree Scores] \label{thm:degree_stability}
    When $\sigma(\mA) = \vs = \mA^T\ve$, the vector $\vs_0 = d\ve$ is a root of $\vf$, where $d = \frac{m}{n}$, and is the only egalitarian root.
    Furthermore, $\vs_0$ is linearly stable in the long-memory limit if and only if $\mM(\vs_0;\vbeta)$ has eigenvalues strictly smaller than $\frac{1}{m}$. 
\end{thm}
\begin{proof}
    We first derive the functional form of $\vf$. 
    We can write
    \begin{align*}
        \E[\vs\tplusone|\mA\t] &= \E[\mA\tplusone|\mA\t]^T\ve \\
                               &= \lambda\mA\t\ve + (1-\lambda)\E[\mDelta\t]^T\ve \\ 
                               &= \lambda\mA\t\ve + (1-\lambda)mn^{-1}\mG\t^T\ve\;.
    \end{align*}
    Inserting this expression into \eqref{eq:limit_f}, and recognizing $n^{-1}\mG\t\ve = \vgamma\t$, we have 
    \begin{align*}
        \vf(\vs) = mn^{-1}\E[\mG]\ve-\mA\t\ve = m\vgamma-\vs\;.
    \end{align*}
    We can now check that $\vs_0$ is indeed the unique egalitarian root of $\vf$. 
    Suppose that $\vs = s\ve$ for some scalar $s$.  
    Then, 
    \begin{align*}
        \vf(\vs) = m\vgamma(\vs) - \vs = (mn^{-1} - s)\ve\;,
    \end{align*}
    which is only equal to zero when $s = \frac{m}{n}$, as needed. 

    Now computing derivatives, we have 
    \begin{align*}
        \dds{\vf(\vs)} = m\mM(\vs;\vbeta) - \mI\;.
    \end{align*}
    This matrix has strictly negative eigenvalues provided that the eigenvalues of $\mM(\vs_\vzero;\vbeta)$ are strictly smaller than $\frac{1}{m}$, completing the proof. 
\end{proof}
\begin{cor}
    Using the Root-Degree score function, $\vs_0 = \frac{m}{n}\ve$ is a linearly stable fixed point of $\vf$ if and only if $\beta < 2\sqrt{\frac{n}{m}}$. 
\end{cor}
\begin{proof}
    It is convenient to treat the operation of taking the square root as part of the feature map, rather than part of the score function. 
    We therefore suppose that $s_j$ is the in-degree of node $j$ and that  $\phi_{j}(\vs) = \sqrt{s_j}$. 
    Computing from \eqref{eq:dgamma}, we obtain 
    \begin{align*}
        \mM(\vs_0;\vbeta) = \frac{1}{2}\frac{n^{-1}}{\sqrt{d}}\beta(\mI - n^{-1}\mE)\;.
    \end{align*}
    This matrix again has a zero eigenvalue associated with the direction $\ve$. 
    For any direction $\vv \perp \ve$, there is an eigenvalue $\frac{1}{2}\frac{n^{-1}}{\sqrt{d}}\beta$. 
    From \Cref{thm:degree_stability},  $\vs_0$ will be linearly stable provided that 
    \begin{align*}
        \frac{1}{m}> \frac{1}{2}\frac{n^{-1}}{\sqrt{d}}\beta\;.
    \end{align*}
    or 
    \begin{align*}
        \beta < 2\sqrt{d}\frac{n}{m} = 2\sqrt{\frac{n}{m}}\;,
    \end{align*}
    as required. 
\end{proof}

\subsection{PageRank Scores}

The PageRank score~\cite{brin1998anatomy,page1999pagerank} is the solution $\vs$ of the linear system 
\begin{align}
    \left[\alpha \mA^T(\mD^o)^{-1} + (1-\alpha)n^{-1}\mE\right]\vs = \vs\;, \label{eq:PageRank}
\end{align}
where $\mD^o = \diag(\mA\ve)$. 
The Perron-Frobenius Theorem \cite{horn2012matrix} ensures that $\vs$ is strictly positive entrywise.
We assume $\vs$ to be normalized so that $\vs^T\ve = n$, which is contrary to the usual normalization $\vs^T\ve = 1$. 
This choice amounts to a rescaling of the parameters $\vbeta$, and does not otherwise impact the analysis. 

In the case of PageRank, it is difficult to derive a result for general features and we therefore work directly with the PageRank model with linear features. 
\begin{thm}
    The vector $\vs_0 = \ve$ is the unique egalitarian root of $\vf$ 
    under PageRank scores. 
    In the PageRank-Linear model, the egalitarian root is linearly stable if and only if $\beta < \frac{1}{\alpha}$. 
\end{thm}
\begin{proof}
    Uniqueness is a direct consequence of normalization: if $\vs = \theta \ve$ and $\ve^T\vs = n$, then we must have $\theta = 1$. 

    We next obtain a necessary condition describing roots of $\vf$. 
    We start with a useful simplification. 
    At any fixed point of $\vf$, we must have $\mD^o = m\mI$. 
    This is because, at any such fixed point, we must have $\mA = m\mG$, and $n\mG$ is row-stochastic. 
    For the purposes of analysis in the long-memory limit, we can therefore consider $\vs$ to be defined by the simplified equation
    \begin{align}
        \left[\alpha m^{-1}n\mA^T + (1-\alpha) n^{-1}\mE \right]\vs = \vs\;. \label{eq:PageRank_simple}
    \end{align}
    In the next time step, we will have 
    \begin{align*}
        \left[\alpha m^{-1}n (\mA^T + \delta \mA^T) + (1-\alpha)n^{-1}\mE\right](\vs + \delta \vs) = \vs + \delta \vs\;.
    \end{align*}
    Expanding and canceling yields
    \begin{align*}
        \left[\alpha m^{-1}n\mA^T + (1-\alpha)n^{-1}\mE\right]\delta \vs + \alpha m^{-1}n(\delta \mA^T)\vs + o(1-\lambda)= \delta \vs\;.
    \end{align*}
    The term $o(1-\lambda)$ includes terms involving the product $(\delta \mA^T)( \delta\vs)$, and relies on the fact that $\delta \vs$ is a smooth function of $\mA$. 
    Rearranging and dropping the asymptotic term, we obtain, in the long memory limit,
    \begin{align}
        \left[ \mI - \alpha m^{-1}n\mA^T - (1-\alpha)n^{-1}\mE \right]\delta \vs = \alpha m^{-1}n(\delta \mA^T)\vs\;. \label{eq:PR_intermediate}
    \end{align}
    This expression gives an implicit representation of $\vf$ via the relation $\vf(\vs, \mA) = \lim_{\lambda \rightarrow 1} \frac{\E[\delta \vs]}{1-\lambda}$. 
    We can therefore enforce $\vf(\vs, \mA) = \vzero$ by setting $\E[\delta \vs] = \vzero$, obtaining the  necessary condition $\E[\delta \mA^T]\vs = \vzero$ for roots of $\vf$. 
    Expanding this condition yields, 
    \begin{align*}
        \vzero = \E[\delta \mA^T]\vs = (1-\lambda)(\mG^T - \mA^T)\vs \;.
    \end{align*}
    Inserting \eqref{eq:PageRank_simple} and rearranging yields the nonlinear system 
    \begin{align}
        \left[\mG^T + \alpha^{-1}(1-\alpha)n^{-2}\mE\right]\vs = \alpha^{-1}n^{-1}\vs\;. \label{eq:PageRank_equilibria}
    \end{align}
    The largest eigenvalue of the matrix on the lefthand side is $\alpha^{-1}n^{-1}$. 
    This allows us to numerically solve \eqref{eq:PageRank_equilibria} iteratively, by alternating between solving for $\vs$ via a standard eigenvalue solver and updating $\mG$ with the new value of $\vs$.
    This is the method implemented in the accompanying software and used to generate equilibria in Fig.~3.  
    
    In order to derive the linear stability criterion, we  divide both sides of \eqref{eq:PR_intermediate} by $ 1-\lambda$ and differentiate with respect to $\vs$, obtaining
    \begin{align*}
        \left[ \mI - \alpha m^{-1} n\mA^T - (1-\alpha)n^{-1}\mE \right]\mJ(\vs) = \alpha m^{-1} n \dds{}\left[\mG^T\vs - \mA^T\vs\right]\;.
    \end{align*} 
    After inserting \eqref{eq:PageRank_simple} and simplifying, we have 
    \begin{align*}
        \left[ \mI - \alpha m^{-1} n\mA^T - (1-\alpha)n^{-1}\mE \right]\mJ(\vs) &= \alpha m^{-1} n\dds{}\left[\mG^T\vs - \alpha^{-1}mn^{-1}\vs + \alpha^{-1}(1-\alpha)mn^{-2}\mE\vs \right] \\ 
        &= \alpha m^{-1} n\dds{}\left[\mG^T\vs - \alpha^{-1}mn^{-1}\vs \right]\;.
    \end{align*}
    The second line follows from the normalization of $\vs$, which implies that $\mE\vs=n\ve$, a constant vector which does not depend on $\vs$. 
    Differentiating the righthand side then yields  
    \begin{align*}
        \left[ \mI - \alpha m^{-1} n\mA^T - (1-\alpha)n^{-1}\mE \right]\mJ(\vs) &= \alpha m^{-1} n \left[\mG^T + (\ve^T\vs)mn^{-1}\dds{\vgamma}\right] - \mI \\ 
        &= \alpha m^{-1} n \left[\mG^T + m\dds{\vgamma}\right] - \mI\;. 
    \end{align*}
    Evaluated at the egalitarian solution $\vs_0 = \ve$, this becomes
    \begin{align*}
        \left[ \mI - \alpha m^{-1} n\mA^T - (1-\alpha)n^{-1}\mE \right]\mJ(\vs_0) = \alpha m^{-1}n^{-1} \mE + \alpha \mM(\vs_0;\vbeta) - \mI\;. 
    \end{align*}

    To complete the argument, we note that, at the egalitarian solution of our model dynamics, $\mA = n^{-2}\mE$. 
    Inserting and simplifying, we have 
    \begin{align*}
        \left[ \mI - \alpha m^{-1} n^{-1}\mE \right]\mJ(\vs_0) = \alpha n^{-1} m^{-1}\mE + \alpha n \mM(\vs_0;\vbeta) - \mI\;. 
    \end{align*}
    Provided that $\alpha < 1$, the premultiplying matrix on the lefthand side is invertible, and $\left[ \mI - \alpha m^{-1} n^{-1}\mE \right]^{-1} = \mI + \alpha(m-\alpha)^{-1}n^{-1}\mE$. 
    This matrix has a single eigenvalue $1 + \alpha(m-\alpha)^{-1}$ with eigenvector $\ve$, and additional eigenvalues equal to unity in orthogonal directions. 
    We then have 
    \begin{align*}
        \mJ(\vs_0) = \alpha m^{-1}(1+\alpha(m-\alpha)^{-1})\mE +  \alpha n\left[\mI + \alpha(m-\alpha)^{-1}n^{-1}\mE \right]\mM(\vs_0;\vbeta) - \mI\;.
    \end{align*}
    In the PageRank-Linear model, $\mM(\vs_0;\vbeta) = \beta n^{-1}(\mI -n^{-1}\mE)$, and we therefore have 
    \begin{align*}
        \mJ(\vs_0) = \alpha m^{-1}(1+\alpha(m-\alpha)^{-1})\mE + \alpha\beta  \left[\mI + \alpha(m-\alpha)^{-1}n^{-1}\mE \right](\mI -n^{-1}\mE) - \mI\;.
    \end{align*}
    We can now read off the eigenvalues of $\mJ(\vs_0)$ analytically. 
    The eigenvector $\ve$ has eigenvalue $-1$, while any vector orthogonal to $\ve$ has eigenvalue $\alpha \beta - 1$.
    This latter eigenvalue is strictly negative if and only if $\beta < \frac{1}{\alpha}$, as was to be shown. 
\end{proof}

\subsection{SpringRank Scores}

We return to the general formalism of score functions and features introduced at the beginning of this section. 

A SpringRank vector $\vs$ for a matrix $\mA$ with regularization $\alpha\in \R$ is a solution to the linear system 
    \begin{align}
        \left[\mD^i + \mD^o - (\mA + \mA^T) + \alpha \mI\right]\vs = \vd^i - \vd^o.  
    \end{align} 
    where, $\vd^i = \ve^T\mA$, $\vd^o = \mA^T\ve$, $\mD^i = \diag(\vd^i)$, and $\mD^o = \diag(\vd^o)$.  
   When $\alpha > 0$, \eqref{eq:springrank} is invertible and $\vs$ is therefore unique. 
    Thus, throughout this section we will assume that $\alpha > 0$, and correspondingly refer to $\vs$ as ``the'' SpringRank vector of $\mA$. 
    It is convenient to define $\mL_\alpha = \mD^i + \mD^o - (\mA + \mA^T) + \alpha \mI$ and $\mLambda = \mD^i - \mD^o$, in which case the SpringRank relation reads $\mL_\alpha\vs = \mLambda \ve$. 

\begin{thm}[Stable Egalitarianism with SpringRank Scores] \label{thm:springrank_stability}
    When $\sigma$ is the SpringRank map, the vector $\vs_0 = \vzero$ is a fixed point of $\vf$, and is the only egalitarian fixed point of the dynamics. 
    This fixed point is linearly stable in the long-memory limit if and only if the matrix
    \begin{align*}
        \mM(\vzero;\vbeta) - 2n^{-1}(\mI - n^{-1}\mE)
    \end{align*}
    has eigenvalues strictly smaller than $\frac{\alpha n}{m}$.
\end{thm}

We will break the proof into a series of three lemmas. 
The first lemma calculates the analytical form of $\vf$. 
The second shows that $\vs_0=\vzero$ is the unique egalitarian fixed point of the long-memory limiting dynamics $\vf$. 
The third gives the criterion for linear stability. 

\begin{lm}
    The deterministic approximant $\vf$ for the SpringRank vector is given by 
    \begin{align}
         \vf(\vs, \mA) = \vs + \mL_\alpha^{-1}\left(-\alpha\vs  - m\left(n^{-1}\mL_{\mG}\vs - (n^{-1}\ve - \vgamma)\right)\right)\;, \label{eq:f_springrank}
     \end{align}
     where $\mL_{\mG} = \mGamma + n^{-1}\mI - (\mG + \mG^T)$.
      
\end{lm}
\begin{proof}
    
    Let us fix an implicit time step $t$. 
    Here and below, we use the notational template $\delta M = M\tplusone - M\t$ to refer to increments in various quantities under the dynamics \eqref{eq:dynamics}. 
    For example, $\delta\mA = \mA\tplusone - \mA\t$  refers to the increment in $\mA$ under the dynamics. 
    We compute directly
    \begin{align*}
        \delta \mA   &= (\lambda - 1)(\mA - \mDelta) \\ 
        \delta \mD^o &= (\lambda - 1)(\mD^o - \diag (\mDelta \ve)) \\ 
        \delta \mD^i &= (\lambda - 1)(\mD^i - \diag (\mDelta^T\ve))\;.
    \end{align*}
    We can also explicitly write out formulae for the increments in $\mL_\alpha$ and $\mLambda$: 
    \begin{align}
        \delta \mLambda &= \delta \mD^i - \delta \mD^o \nonumber \\
                        &= (\lambda - 1)\left[\mD^i - \mD^o + \diag((\mDelta-\mDelta^T)\ve) \right] \nonumber \\ 
                        &= (\lambda - 1)\left[\mLambda + \diag((\mDelta-\mDelta^T)\ve) \right] \;, \label{eq:dlambda} \\ 
        \delta \mL_\alpha &=  \delta \mD^i + \delta \mD^o - (\delta \mA  +\delta \mA^T)  \nonumber \\ 
                   &= (\lambda - 1)\left[\mD^i + \mD^o - \diag(\mDelta^T\ve + \mDelta\ve) - (\mA + \mA^T) + \mDelta + \mDelta^T  \right] \nonumber \\ 
                   &= (\lambda - 1)\left[\mL - \diag(\mDelta^T\ve + \mDelta\ve) + \mDelta + \mDelta^T  \right] \nonumber \\ 
                   &\triangleq (\lambda - 1)\left[\mL - \mL_{\mDelta} \right]\;, \label{eq:dL}
    \end{align}
    where we have given a name to the Laplacian $\mL_{\mDelta} = \diag(\mDelta^T\ve + \mDelta\ve) - \mDelta^T - \mDelta$ of $\mDelta$. 
    Note that $\delta\mL_\alpha$ does not depend on $\alpha$, and we therefore simply write $\delta\mL = \delta\mL_\alpha$.

    We can now formulate a simple condition for equilibrium in expectation. 
    We have 
    \begin{align*}
        (\mL_\alpha + \delta\mL) (\vs + \delta\vs)
        = (\mLambda  + \delta\mLambda)\ve\;.
    \end{align*}
    Subtracting the SpringRank relation $\mL_\alpha\vs = \mLambda \ve$ from each side of this expression, we obtain 
    \begin{align*}
        (\mL_\alpha + \delta\mL)\delta\vs
        = (\delta\mLambda)\ve - (\delta\mL) \vs \;. 
    \end{align*}
    Since $\delta \mL = O(1-\lambda)$, the lefthand matrix is invertible in for small $\lambda$ provided that $\alpha > 0$.  
    We therefore obtain
    \begin{align*}
        \delta \vs &= \left(\mL_\alpha^{-1} + O(1-\lambda) \right)((\delta\mLambda)\ve - (\delta\mL) \vs) \\ 
        &= \mL_\alpha^{-1} ((\delta\mLambda)\ve - (\delta\mL )\vs) + O((1-\lambda)^2)\;.
    \end{align*} 
    The term $O((1-\lambda)^2)$ arises from the product of $O(1-\lambda)$ and the copy of $(\lambda - 1)$ within $\delta\mLambda$ and $\delta\mL$. 
    Taking expectations, 
    \begin{align*}
        \E[\delta\vs] = \mL_\alpha^{-1} (\E[\delta\mLambda]\ve - \E[\delta\mL]\vs) + O((1-\lambda)^2)\;.
    \end{align*}
    We next insert the expressions \eqref{eq:dlambda} and \eqref{eq:dL} and use the fact that $\E[\mDelta] = m\mG$. 
    This gives 
    \begin{align*}
        \E[\delta\vs] = (1-\lambda)\mL_\alpha^{-1} \left(\left[\mL - m\mL_{\mG}\right]\vs - \left[\mLambda + m\cdot\diag((\mG - \mG^T)\ve)\right]\ve\right) + O((1-\lambda)^2)\;.
    \end{align*}
    We can simplify this expression by recalling that $(\mL + \alpha \mI)\vs = \mLambda \ve$ by definition, as well as the identities $\mG\ve = n^{-1}\ve$ and $\mG^T\ve = \vgamma$. 
    Inserting these identities and simplifying yields 
    \begin{align*}
        &= (1-\lambda)\mL_\alpha^{-1}\left(-\alpha\vs  - m\left(\mL_{\mG}\vs + (n^{-1}\ve - \vgamma)\right)\right) + O((1-\lambda)^2) \;.
    \end{align*}
    We now construct $\vf$, obtaining
    Since $\E[\delta \vs] = \E[\sigma(\lambda \mA + (1-\lambda)\mDelta)]$, we can write 
    \begin{align*}
        \vf(\vs, \mA) &= \vs + \lim_{\lambda \rightarrow 1}\frac{\E[\delta \vs]}{1-\lambda} \\
                      &= \vs - \mL_\alpha^{-1}\left[\alpha\vs  + m\left(\mL_{\mG}\vs + (n^{-1}\ve - \vgamma)\right)\right]\;, 
    \end{align*}
    as was to be shown. 
\end{proof}

\begin{lm}
    When $\sigma$ is the SpringRank map, the vector $\vs_0 = \vzero$ is a root of $\vf$, and is the only egalitarian fixed point. 
\end{lm}
\begin{proof}
    To show that $\vs_0 = \vzero$ is a fixed point of $\vf$, it suffices to insert this solution into \eqref{eq:f_springrank} and simplify, noting that, when $\vs = \vzero$, $\vgamma = n^{-1}\ve$. 
    To show that it is the unique egalitarian root realizable as a SpringRank score, suppose that $s \ve$ were a SpringRank score for some $s \neq 0$. 
    Inserting this into \eqref{eq:springrank} and using the fact that $\ve$ is a zero eigenvector of the unregularized Laplacian, we would have 
    \begin{align*}
        \alpha s \ve = \vd^i - \vd^o\;.
    \end{align*}
    The total in-degree must equal the total out-degree. 
    Pre-multiplying by $\ve$ therefore zeros out the righthand, leaving: 
    \begin{align*}
        \alpha s \ve^T\ve = \alpha sn = 0\;,
    \end{align*}
    which is a contradiction unless $s = 0$. 
\end{proof}
\begin{lm}
    The egalitarian root $\vs = \vzero$ is a linearly stable root of the SpringRank dynamics in the long-memory limit if and only if the matrix
    \begin{align*}
        \mM(\vzero;\vbeta) - 2n^{-1}(\mI - n^{-1}\mE)
    \end{align*}
    has eigenvalues strictly smaller than $\frac{\alpha}{m}$.
\end{lm}

\begin{proof}
    We need to compute $\mJ(\vs_0)$, the Jacobian matrix of $\vf$ at $\vs_0 = \vzero$. 
    The fixed point will be stable provided that $\mJ(\vs_0)$ has strictly negative eigenvalues.
    To compute this Jacobian, we compute derivatives in \eqref{eq:f_springrank}. 
    Doing so and applying the product rule, we have 
    \begin{align*}
        \dds{f(\vs)} = \mI - \mL_{\alpha}^{-1}\left(\alpha \mI + m\left(n^{-1}\dds{(\mL_\mG\vs)}  - \dds{\vgamma}\right)\right)\;.
    \end{align*}
    We calculate $\dds{\mL_{\mG}}$ in \Cref{eq:L_G}, now obtaining
    \begin{align*}
        \dds{f(\vs)} = \mI - \mL_{\alpha}^{-1}\left(\alpha \mI + m\left(n^{-1}\left[\mL_{\mG} + \mSigma \dds{\vgamma} - \dds{\vgamma}(\mS^T + (\ve^T\vs)\mI)\right]  - \dds{\vgamma}\right)\right)\;.
    \end{align*}

    Evaluating this expression at $\vs = \vzero$, we have 
    \begin{align*}
        \mJ(\vzero) =  - \mL_{\alpha}^{-1}\left(\alpha \mI + m\left(n^{-1}\mL_{\mG} - \dds{\vgamma(\vzero)}\right)\right)\;,
    \end{align*}
    where $\mL_\mG$ must also be evaluated at $\vs = \vzero$. 
    We have $\mG(\vzero) = n^{-1}\mE$, which implies $\mL_{\mG} = 2(\mI - n^{-1}\mE)$. 
    We insert this expression and the formula for $\dds{\vgamma}$ given in \eqref{eq:dgamma}, obtaining
    \begin{align*}
        \mJ(\vzero) = - \mL_{\alpha}^{-1}\left[\alpha \mI + mn^{-1}(\mI - n^{-1}\mE)\left(2\mI - \sum_{i = 1}^n\sum_{\ell = 1}^k \beta_\ell \dds{\phi_i^{\ell}(\vs_0)}\right)\right]\;. 
    \end{align*}
    Since $\mL_{\alpha}$ is symmetric and positive-definite, $\mL_{\alpha}^{-1}$ is as well. 
    The stability of the egalitarian fixed point is therefore determined by the eigenvalues of the matrix inside the brackets. 
    Multiplying by $nm^{-1}$, we find that a necessary and sufficient condition is that the matrix 
    \begin{align*}
        (\mI - n^{-1}\mE)\left(2\mI - \sum_{i = 1}^n\sum_{\ell = 1}^k \beta_\ell \dds{\phi_i^{\ell}(\vs_0)}\right) = \mM(\vzero;\vbeta) - 2n^{-1}(\mI - n^{-1}\mE)
    \end{align*}
    have eigenvalues no larger than $\frac{\alpha}{m}$, completing the proof. 
\end{proof}

\begin{cor}
    In the SpringRank-Linear model, $\vs_0 = \vzero$ is a linearly stable fixed point of $\vf$ if and only if $\beta < 2 + \frac{\alpha n}{m}$. 
\end{cor}
\begin{proof}
    It suffices to specialize \Cref{thm:springrank_stability} to the case of linear features. 
    In particular, we have $\mM(\vzero; \beta) = \beta n^{-1}(\mI - n^{-1}\mE)$. 
    We therefore require that the matrix 
    \begin{align*}
        \beta n^{-1}(\mI - n^{-1}\mE) - 2n^{-1}(\mI - n^{-1}\mE) = n^{-1}(\beta - 2)(\mI - n^{-1}\mE)
    \end{align*}
    have eigenvalues smaller than $\frac{\alpha}{m}$. 
    We can compute the eigenvalues of this matrix analytically -- there is a zero eigenvalue corresponding to the vector $\ve$. 
    Then, any vector $\vv \perp \ve$ is also an eigenvector with eigenvalue $n^{-1}(\beta - 2)$. 
    We therefore require $n^{-1}(\beta - 2) < \frac{\alpha }{m}$, or $\beta < 2 + \frac{\alpha n}{m}$, completing the argument. 
\end{proof}

\begin{lm}
    We have 
    \begin{align}
        \dds{\mL_{\mG}\vs} &= \mL_{\mG} + \mSigma \dds{\vgamma} - \dds{\vgamma}(\mS^T + (\ve^T\vs)\mI)\;.\label{eq:L_G}
    \end{align}
\end{lm}
\begin{proof}
    We first compute the derivatives $\dds{(\mG\vs)}$ and $\dds{(\mG^T\vs)}$. 
    The $i$th component of $\mG\vs$ is $v_{i} = \sum_{j}\gamma_js_j$. 
    The product rule for scalar functions of vectors gives the $i$th row of the derivative: 
    \begin{align*}
        \dds{\mG\vs_i} &=\sum_j \gamma_j\ve_j + \sum_j s_j\dds{\gamma_j} = \vgamma + \sum_j s_j\dds{\gamma_j}\;.
    \end{align*}
    Written in matrix notation, the first term is $\mG$. 
    To write the second term in matrix form, note that we need to multiply $\dds{\vgamma}$ by the matrix each of whose columns is a copy of $\vs$. 
    This matrix is $\mS^T$. 
    We therefore obtain 
    \begin{align*}
        \dds{(\mG\vs)} = \mG +  \dds{\vgamma}\mS^T\quad \;.
    \end{align*}
    To compute the second derivative, note that $\mG^T\vs = \vgamma(\ve^T\vs)$, with $i$th component $\gamma_i \ve^T\vs$. 
    Using the product rule for scalar functions of vectors, we have
    \begin{align*}
        \dds{}\gamma_i \ve^T\vs = \gamma_i \ve + (\ve^T\vs)\dds{\gamma_i} . 
    \end{align*}
    The first term will become the matrix whose $i$th row is $\gamma_i$, i.e. $\mG^T$. 
    This yields 
    \begin{align*}
        \dds{(\mG^T\vs)}= \mG^T + (\ve^T\vs)\dds{\vgamma}\;.
    \end{align*}
    Combining these expressions yields our formula for $\dds{\mL_{\mG}\vs}$: 
    \begin{align*}
        \dds{\mL_{\mG}\vs} &= \dds{}\left[\mGamma\vs + \vs - \mG\vs - \mG^T\vs \right]\; \nonumber\\
        &= \mGamma + \mSigma \dds{\vgamma} + \mI - \left( \mG + \dds{\vgamma}\mS^T + \mG^T + (\ve^T\vs)\dds{\vgamma}\right) \nonumber \\ 
        &= \mL_{\mG} + \mSigma \dds{\vgamma} - \dds{\vgamma}(\mS^T + (\ve^T\vs)\mI)\;,
    \end{align*}
    as was to be shown. 
\end{proof}

\section{Parameter Estimation} \label{sec:inference}

    Throughout this section, we use the shorthand $\{\mA\t\} = \{\mA\t\}_{t = 0}^{\tau}$ to refer to temporal sequences of matrices up to fixed time $\tau$.  
	We now describe a simple maximum-likelihood model for learning the parameter $\vbeta$ from a sequence of observations $\{\mDelta\t\}$. 
	By construction, $\mDelta(\tau)$ depends on the sequence of state matrices $\{\mA\t\}$ only through the most recent state $\mA(\tau)$.
	We may therefore factor the probability of observing the data given a set of undetermined parameters as: 
	\begin{align*}
		\prob(\{\mDelta\t\};\mA(0), \lambda, \vbeta) = \prod_{t=0}^{\tau} \prob(\mDelta\t;\mA(t), \vbeta)\;.
	\end{align*}
	While the parameter $\lambda$ has disappeared from the righthand side, this expression is nevertheless implicitly a function of $\lambda$ since the value of $\mA(\tau)$ given $\mA(\tau-1)$ and $\mDelta(\tau-1)$ depends on $\lambda$.

	Let us write out a typical factor on the righthand side. 
	Let $\vk_i = \mDelta_{i\cdot}$, and let $K_i = \ve^T\vk_i$. 
	Then,  
	\begin{align*}
		 \prob(\mDelta\t;\mA(\tau), \vbeta) = \prod_{i = 1}^n\left( \frac{K_i}{\prod_{j = 1}^n k_{ij}!} \prod_{j = 1}^n \left(\gamma_{ij}\t\right)^{k_{ij}}\right)\;.
	\end{align*} 
	Taking logarithms and collecting terms that do not depend on $\vbeta$ or $\lambda$ into a constant $C\t$, we obtain 
	\begin{align*}
		\log \prob(\mDelta\t;\mA(t), \vbeta) = \sum_{i = 1}^n \sum_{j = 1}^n k_{ij}\t \log \gamma_{ij}\t + C\t. 
	\end{align*}
	The log-likelihood of the full sequence is then 
	\begin{align*}
		\mathcal{L}(\lambda, \vbeta; \{\mDelta\t\}, \mA(0)) \triangleq \log \prob(\{\mDelta\t\};\mA(0), \lambda, \vbeta) = \sum_{t = 0}^\tau \sum_{i = 1}^n \sum_{j = 1}^n k_{ij}\t \log \gamma_{ij}\t + C, 
	\end{align*}
	where $C = \sum_{t = 0}^\tau C\t$.
	The dependence on $\vbeta$ appears through $\gamma_{ij}$.

	The maximum likelihood approach encourages us to choose as parameter estimates $\hat{\lambda}$ and $\hat{\vbeta}$ the values
	\begin{align}
		\hat{\lambda}, \hat{\vbeta} = \argmax_{\lambda, \vbeta} \mathcal{L}(\lambda, \vbeta; \{\mDelta\t\}, \mA(0))\;. \label{eq:ML}
	\end{align}
	Standard theory of maximum likelihood in exponential families implies that $\mathcal{L}$ is convex in $\vbeta$ for any fixed $\lambda$. 
	This implies that, when $\hat{\lambda}$ is known, we can solve for $\hat{\vbeta}$ via standard first- or second-order optimization methods. 
	Let $\mathcal{L}^*(\lambda; \{\mDelta\t\}, \mA(0))$ be the optimized loglikelihood for fixed $\lambda$. 
	We then complete the maximum likelihood scheme by optimizing $\mathcal{L}^*$ with respect to $\lambda$, which our accompanying software does via a customized hill-climbing algorithm. 
	In general, $\mathcal{L}^*$ may fail to be convex as a function of $\lambda$, and we therefore perform multiple runs with different initial values of $\lambda$ in order to find the global maximum.

\articleonly{
    \newpage
	\section{Additional Model Traces}
}{
	\newpage
	\section{Additional Model Traces}
}

\articleonly{\begin{figure}[h!]}{
    \vspace{2ex}\begin{figure*}[h!]
    }
\articleonly{\begin{center}
\includegraphics[width=\linewidth]{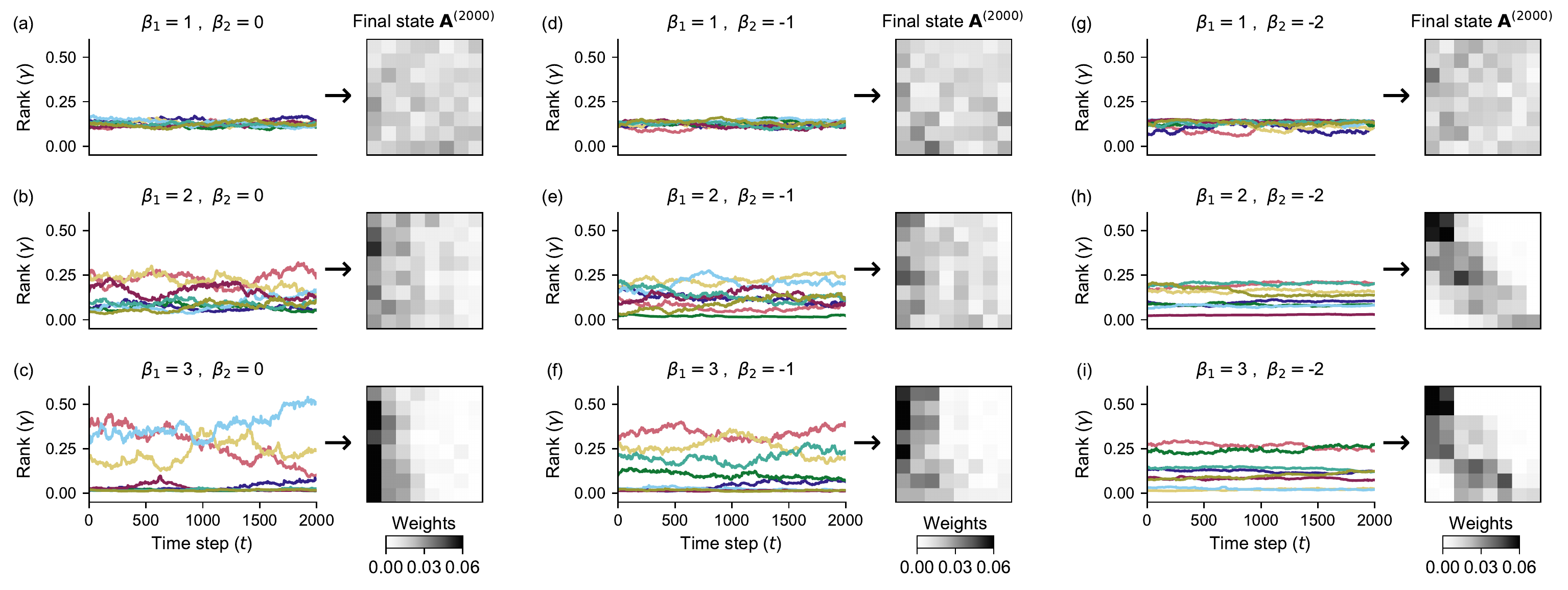}\end{center}\vspace{-15pt}}
{\includegraphics[width=0.85\paperwidth]{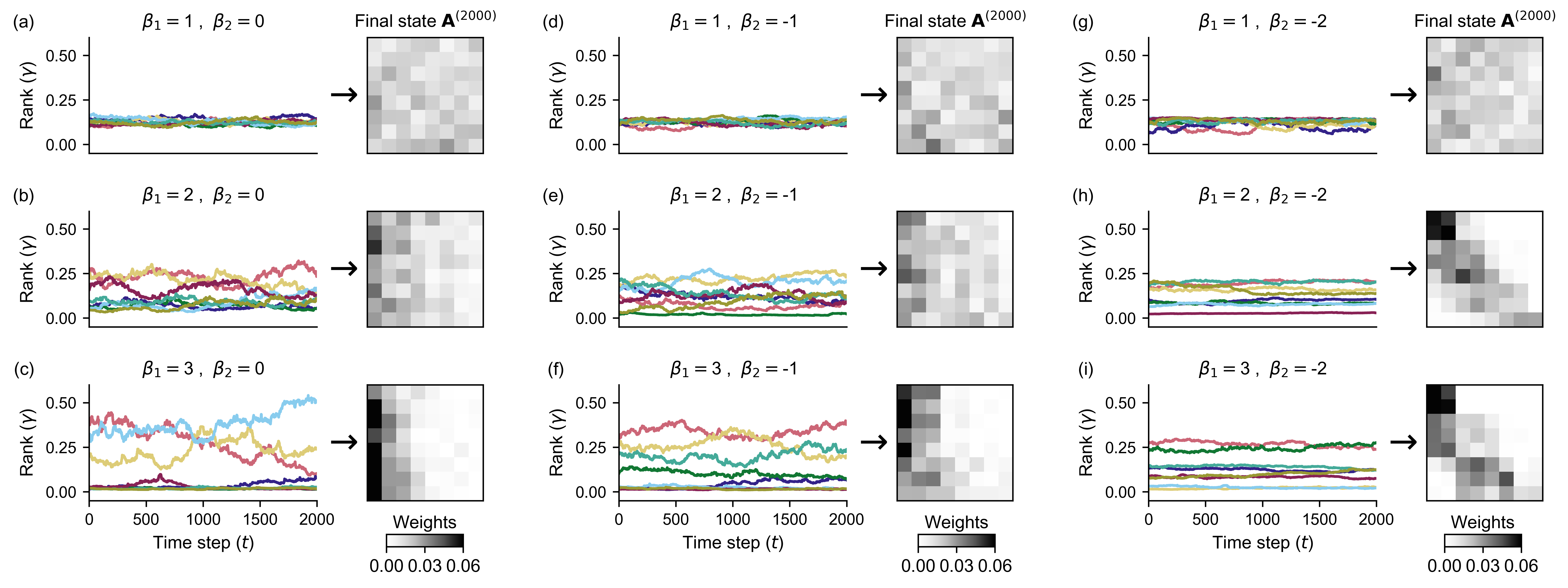}}
\caption{Example dynamics of the model.
 Populations of $n = 8$ agents were simulated for $2000$ time steps using the SpringRank score with linear and quadratic features, varying the preference parameters $\beta_1$ and $\beta_2$ as indicated in the panels. The memory parameter was fixed at $\lambda = 0.995$. In each panel, the plot on the left shows the simulated rank vector $\vgamma$ over time; different colors track the ranks of different agents. The heatmap on the right shows the adjacency matrix $\mA$ at time step $t = 2000$ for the corresponding parameter values.
}
\label{fig:trace-many-springrank}
\articleonly{\vspace{-15pt}\end{figure}}{\end{figure*}}

\articleonly{\begin{figure}[h!]}{
    \vspace{2ex}\begin{figure*}[h!]
    }
\articleonly{\begin{center}
\includegraphics[width=\linewidth]{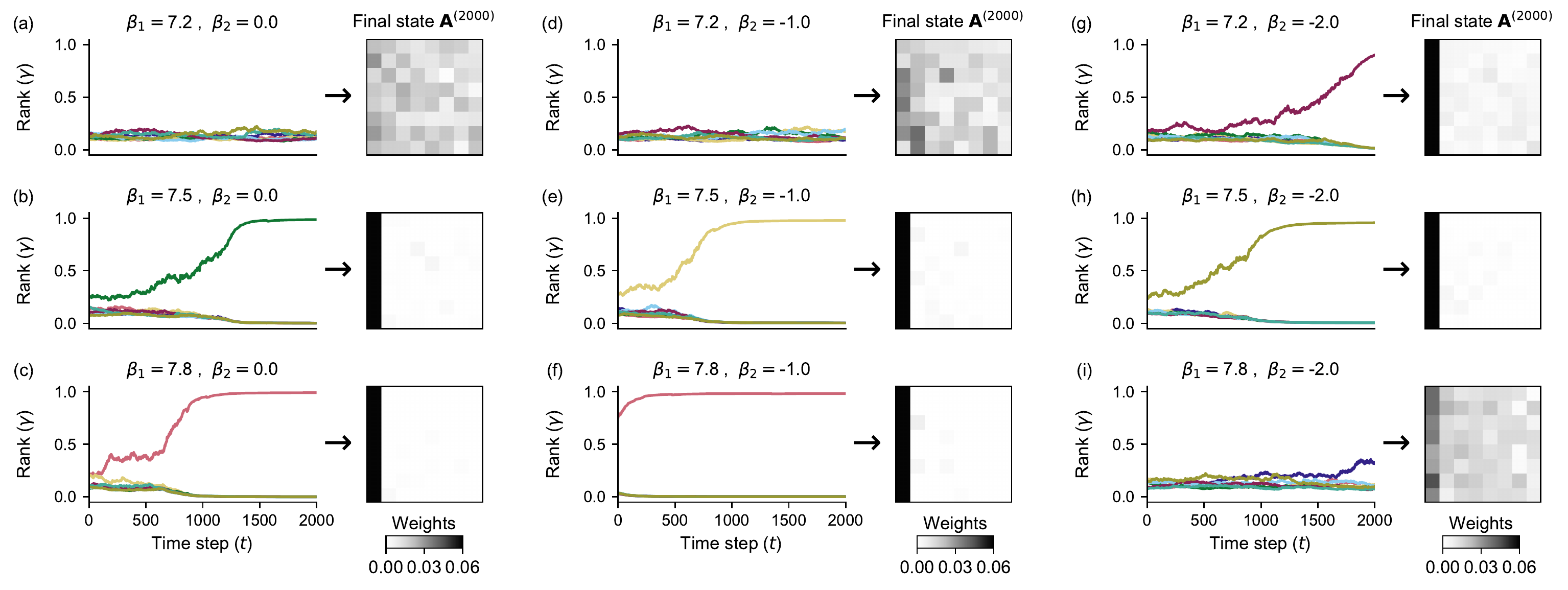}\end{center}\vspace{-15pt}}
{\includegraphics[width=0.85\paperwidth]{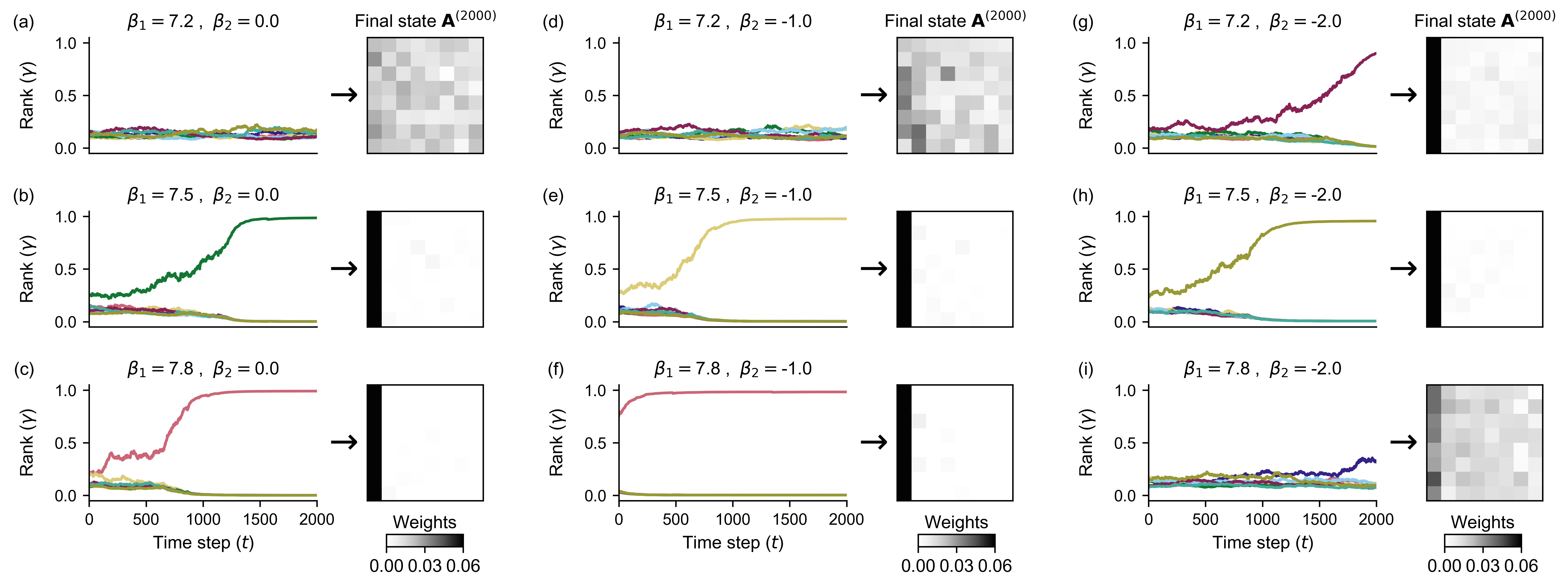}}
\caption{As in \cref{fig:trace-many-springrank}, using the PageRank score function.
}
\label{fig:trace-many-pagerank}
\articleonly{\vspace{-15pt}\end{figure}}{\end{figure*}}

\articleonly{\begin{figure}[h!]}{
		\vspace{2ex}\begin{figure*}[h!]
		}
\articleonly{\begin{center}
    \includegraphics[width=\linewidth]{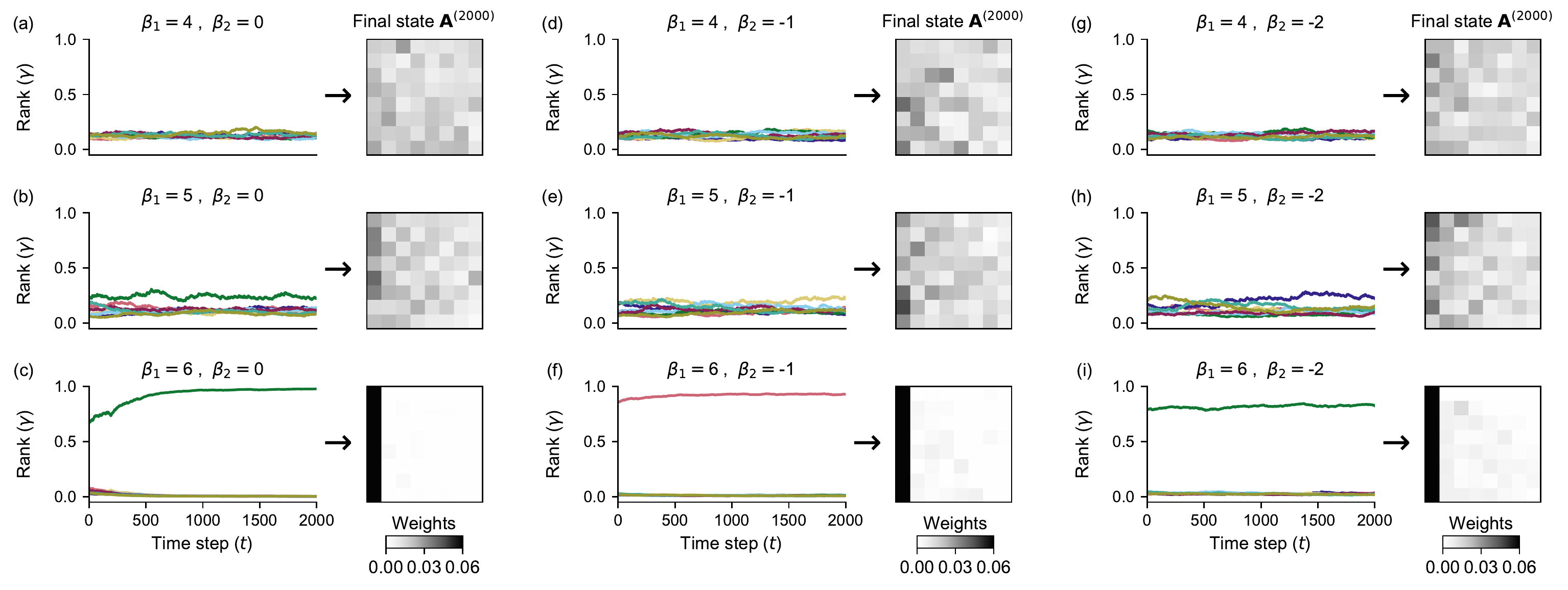}\end{center}\vspace{-15pt}}
    {\includegraphics[width=0.85\paperwidth]{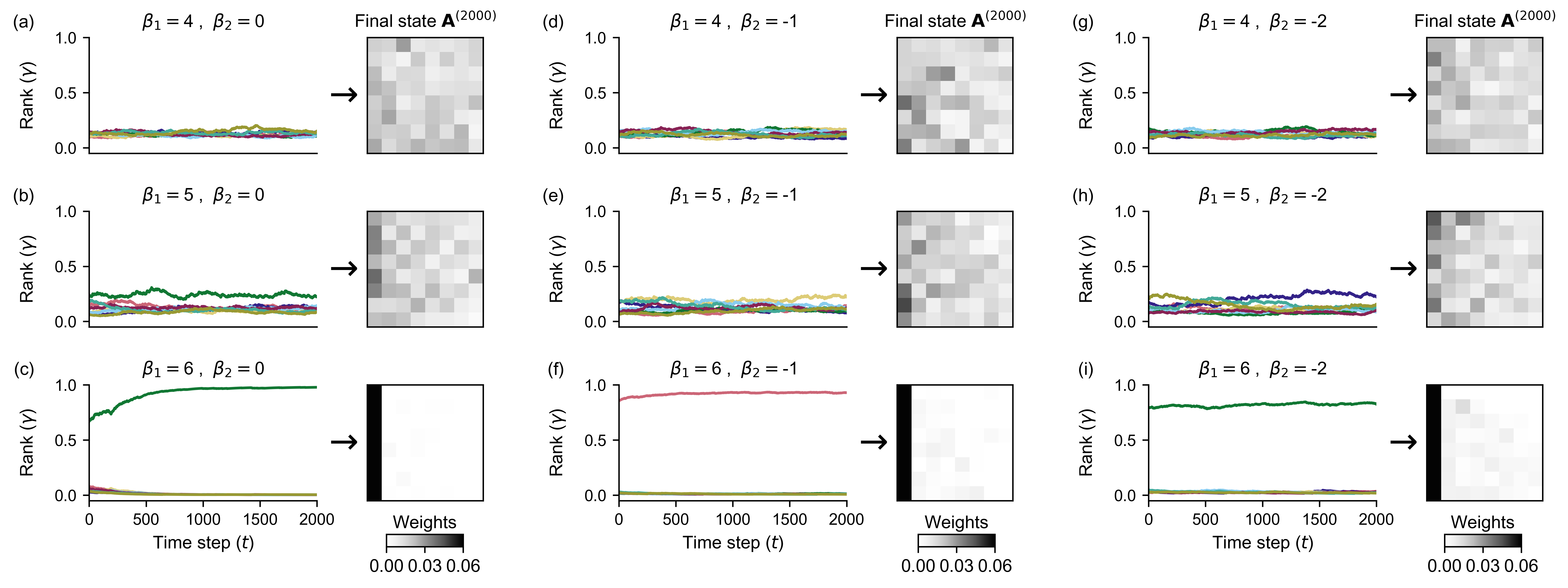}}
    \caption{As in \cref{fig:trace-many-springrank}, using the Root-Degree score function. 
    }
    \label{fig:trace-many-root-degree}
\articleonly{\vspace{-15pt}\end{figure}}{\end{figure*}}

\articleonly{\begin{figure}[h!]}{\vspace{2ex}\begin{tuftefigure}[h!]}
    \articleonly{\begin{center}
        \includegraphics{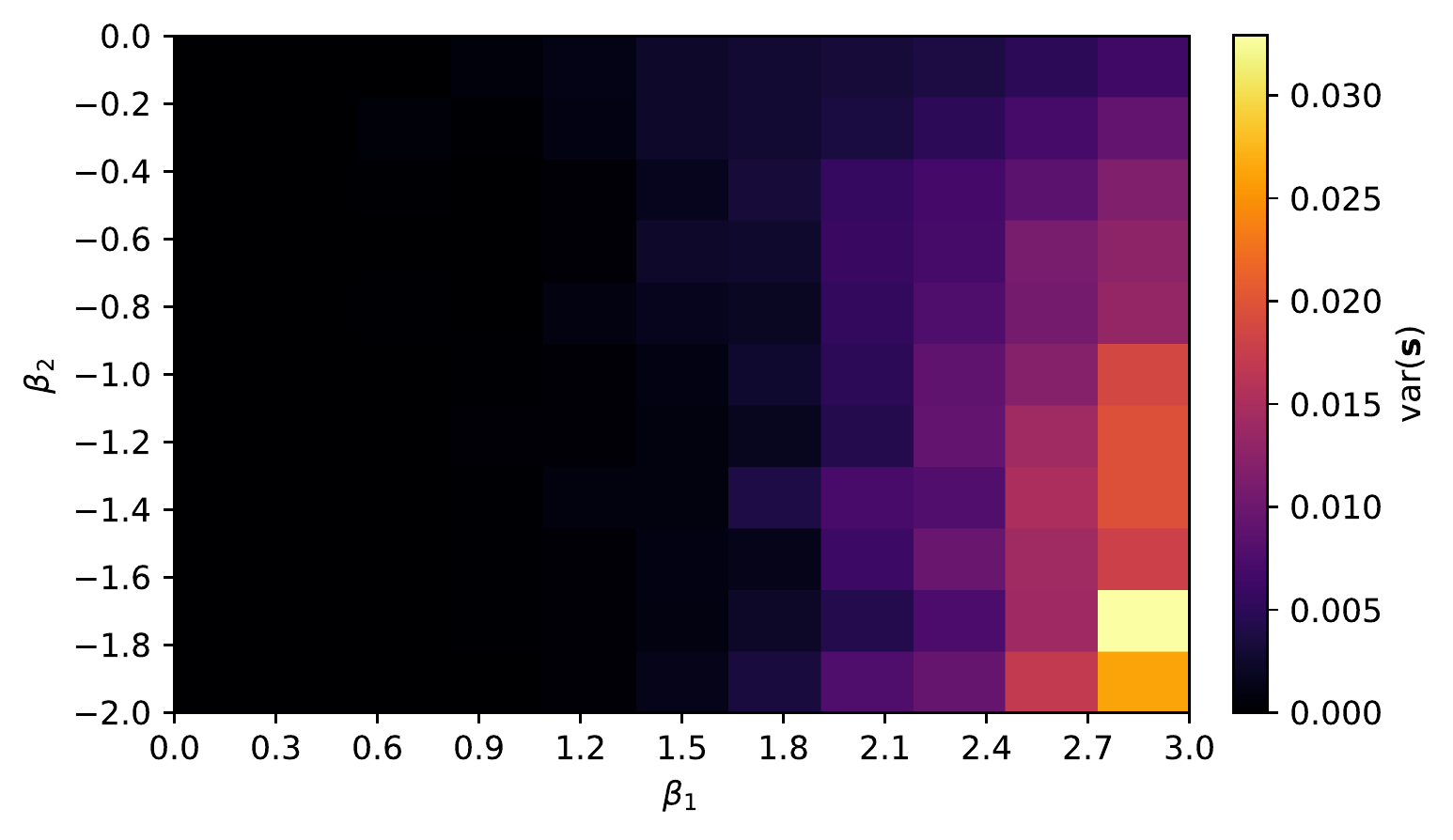}\end{center}}
        {\includegraphics[width=\linewidth]{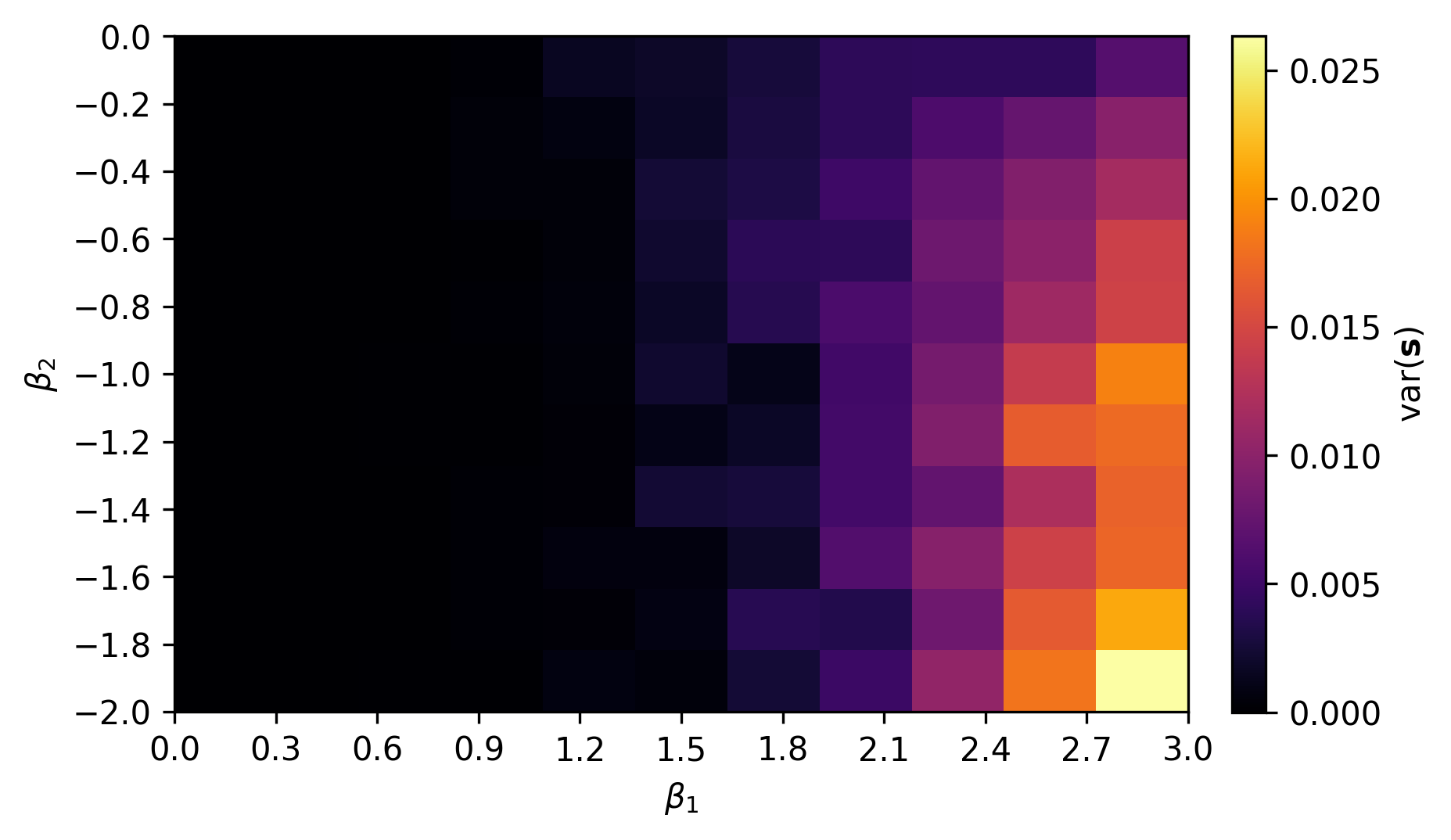}}
    \caption{Plot of the variance in the rank vector $\vs$ over the final 500 iterations of a series of simulations with $n = 8$ and $\lambda = 0.995$ (as in Fig.~2). 
    The parameters $\beta_1$ and $\beta_2$ are allowed to vary. 
    Higher variances correspond to more strongly hierarchical states.  
    }
    \label{fig:beta_heatmap}
    \articleonly{\end{figure}}{\end{tuftefigure}}

\articleonly{\begin{figure}[h!]}{\vspace{2ex}\begin{figure*}[h!]}
\articleonly{\begin{center}
	\includegraphics[width=0.95\linewidth]{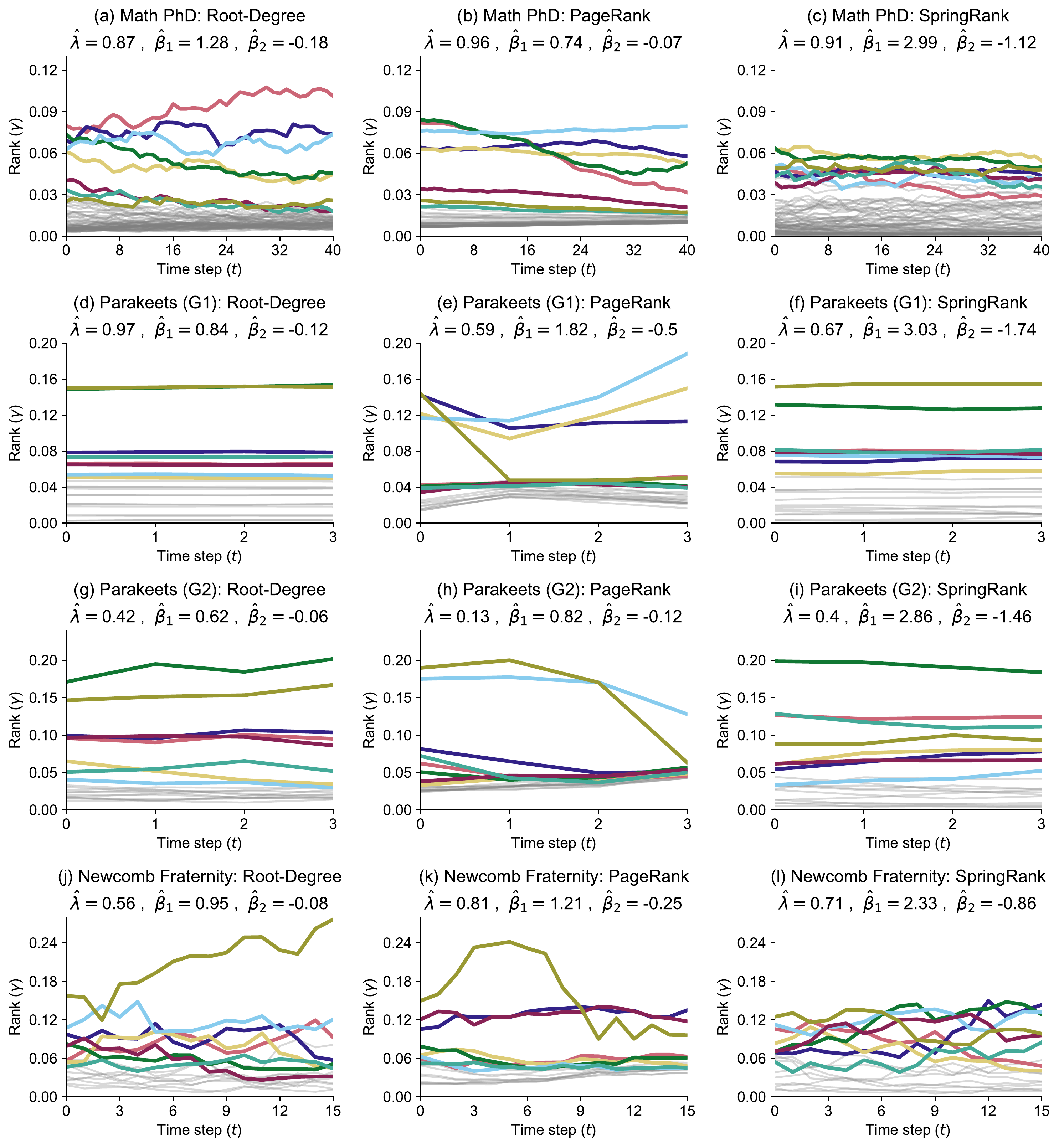}\end{center}}
	{\includegraphics[width=0.85\paperwidth]{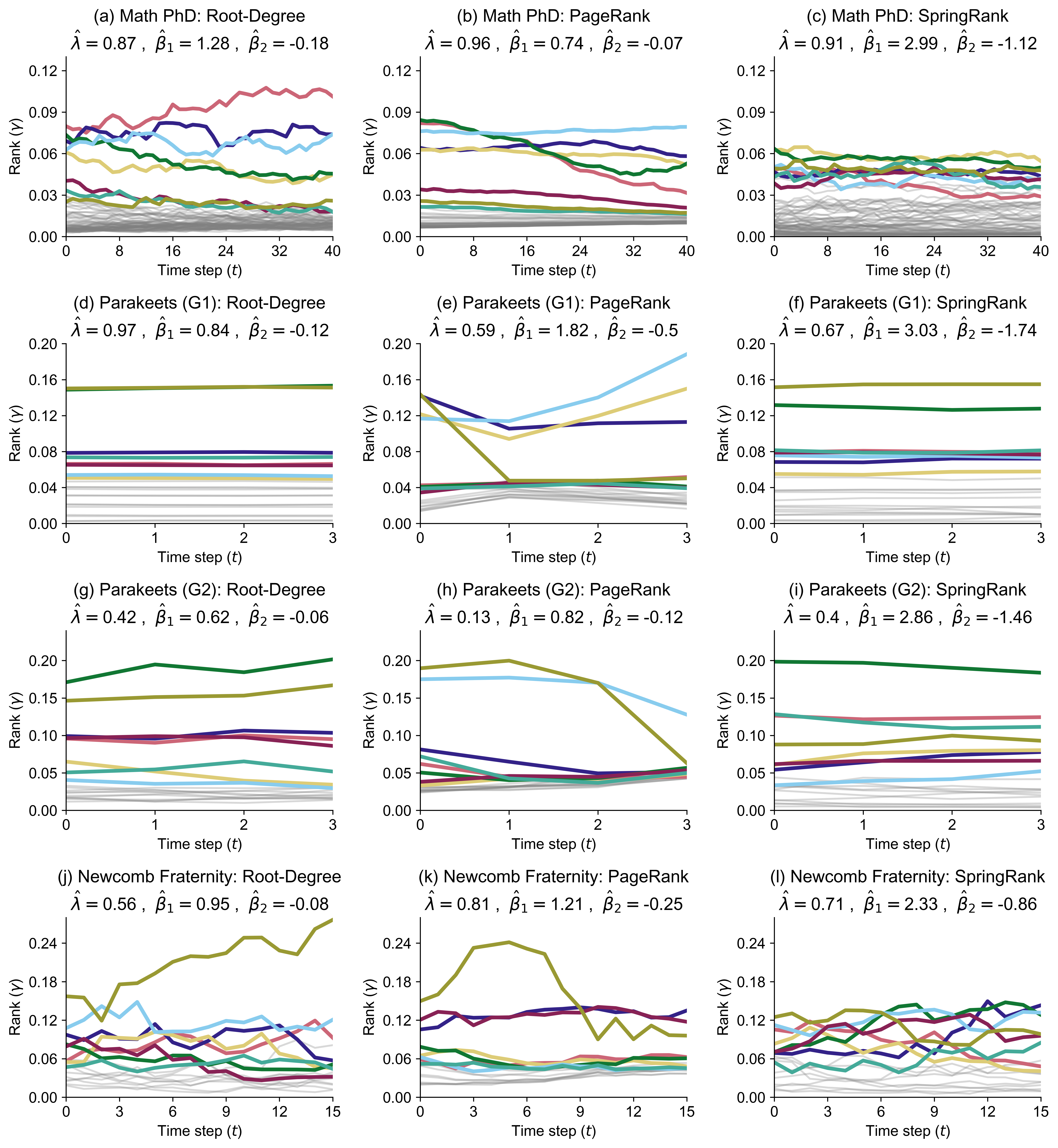}}
\caption{Simulated dynamics of the model using inferred parameters $\hat{\lambda}$, $\hat{\beta}_1$, $\hat{\beta}_2$ in Table 1. 
	The value of $m$ for each row of panels corresponds to the average number of updates per time step in the corresponding data set, indicated in the panel title ($m = 150$ for Math PhD, $m=279$ for Parakeets (G1), $m=320$ for Parakeets (G2), and $m=85$ for Newcomb Fraternity). 
	Furthermore, the simulations in each row were initialized using the network at the relevant initial time step in the corresponding data set: the network of endorsements aggregated up to year 1960 for the Math PhD data set, and the network at time step 0 in each of the Parakeet and Newcomb Fraternity data sets. The traces in color correspond to nodes that rank among the top 8 on average over time; those in light gray track all other nodes.
	Other parameters: $\alpha_p = 0.85$, $\alpha_s = 10^{-8}$.
}
\label{fig:trace-inferred-params}
\articleonly{\end{figure}}{\end{figure*}}

\end{document}